\documentclass{new_tlp}
\usepackage{mathptmx}

\usepackage{listings}
\usepackage{soul}
\usepackage{url}
\usepackage[utf8]{inputenc}
\usepackage[small]{caption}
\usepackage{graphicx}
\usepackage{amsmath}
\usepackage{amssymb}
\usepackage{booktabs}
\usepackage{algorithm}
\usepackage{algorithmic}
\usepackage{mdframed}
\usepackage{xspace}
\usepackage{enumitem}
\usepackage{multirow}

\DeclareMathAlphabet{\mathcal}{OMS}{cmsy}{m}{n}

\newcommand{\rel}[1]{\mathsf{#1}}

\def\HasCol{\rel{HasC}}
\def\Vrtx{\rel{V}}
\def\Col{\rel{Col}}
\def\Edgt{\rel{E_t}}
\def\Edgs{\rel{E_s}}

\def\Colt{\rel{Col_t}}
\def\Emp{\rel{Emp}}
\def\EmpC{\rel{KnownC}}
\def\EmpNew{\rel{EmpC}}
\def\SameC{\rel{SameC}}
\def\Address{\rel{Addr}}
\def\WorkFromHome{\rel{WorkFromHome}}
\def\Ord{\rel{Ord}}
\def\AllOrd{\rel{AllOrd}}
\def\Paid{\rel{Paid}}

\newcommand{\Var}{\mathsf{Var}}
\newcommand{\Const}{\mathsf{Const}}
\newcommand{\Null} {\mathsf{Null}}
\newcommand{\Rel} {\mathsf{Rel}}
\newcommand{\Sch}{\mathsf{S}}
\newcommand{\Src}{\mathsf{S}}
\newcommand{\Trg}{\mathsf{T}}
\def\P{\mathcal{P}}

\def\STTGD{\Sigma_{st}}
\def\Set{\mathcal{S}}
\def\TTGD{\Sigma_t}
\newcommand{\nonull}[1]{#1_{\downarrow}}
\def\des{\langle \Src, \Trg, \STTGD, \TTGD \rangle}
\def\exchoice{ex-choice}
\def\choice{\textit{choice}}
\newcommand{\depg}[1]{\mathsf{dg}(#1)}

\newcommand{\pos}[1]{\mathsf{pos}(#1)}
\newcommand{\norm}[1]{\mathsf{norm}(#1)}
\def\v1{\hspace{-2mm}}

\def\adom{\mathsf{dom}} 

\newcommand{\restr}[2]{#1_{\downarrow #2}}
\def\const{\mathsf{\mathsf{U}}}
\newcommand{\base}[1]{\mathsf{base}(#1)}
\def\var{\mathsf{var}}
\def\exvar{\mathsf{exvar}}
\def\fr{\mathsf{fr}}
\def\nul{{\perp}} 

\newcommand{\ground}[1]{\mathsf{ground}(#1)}
\newcommand{\MM}[1]{\mathsf{MM}(#1)}
\newcommand{\SM}[1]{\mathsf{SM}(#1)}

\def\x{{\bold x}}
\def\y{{\bold y}}
\def\z{{\bold z}}

\def\NP{\text{\rm NP}}
\def\co{\text{\rm co}}

\def\hard{\text{\rm -hard}}
\def\complete{\text{\rm -complete}}
\def\PTIME{\text{\rm PTIME}}

\newcommand{\pw}[1]{\mathsf{pw}(#1)}

\def\ar{\mathit{ar}} 
\newcommand{\conditional}[1]{\dot{#1}}
\def\c-{\, \conditional{-}\, } %
\def\john{\mathsf{john}} %
\def\true{\mathsf{true}} %
\def\false{\mathsf{false}} %

\newcommand{\Cond}[2]{\Phi^{#2}_{#1}}

\def\={\!\mbox{\tiny =}}
\def\<{\langle}
\def\>{\rangle}

\def\ar{ar} 

\newcommand{\Gchase}[3]{\stackrel{#1,#2,#3}{\longrightarrow}}

\def\lprule{rule\xspace}
\def\lprules{rules\xspace}
\def\lprulesep{\text{ :- }}
\newcommand{\LP}[1]{\mathsf{LP}(#1)}
\newcommand{\ED}[1]{\mathsf{ED}(#1)}

\def\body{\mathsf{body}} 
\def\bodyp{\mathsf{body}^+} 
\def\bodyn{\mathsf{body}^-} 
\def\head{\mathsf{head}} 

\newcommand{\cert}[3]{\mathsf{cert}_{#3}(#1,#2)}
\newcommand{\scert}[3]{\mathsf{scert}_{#3}(#1,#2)}
\newcommand{\cans}[3]{\mathsf{cans}_{#3}(#1,#2)}
\newcommand{\ccert}[2]{\mathsf{con\text{-}cert}(#1,#2)}
\newcommand{\scertprob}[1]{\mathsf{SCERT}(#1)}
\newcommand{\supexprob}[1]{\mathsf{EXISTS\text{-}SSOL}(#1)}
\newcommand{\cansprob}[1]{\mathsf{CANS}(#1)}

\newcommand{\sol}[1]{\mathsf{sol}(#1)}
\newcommand{\ssol}[1]{\mathsf{ssol}(#1)}

\newcommand{\catom}[1]{\langle #1 \rangle}

\newcommand{\pushright}[1]{\ifmeasuring@#1\else\omit\hfill$\displaystyle#1$\fi\ignorespaces}

\newcommand{\cinst}[1]{\mathcal{#1}}

\def\bt{\mathbf{t}} 
\def\bu{\mathbf{u}}
\def\bx{\mathbf{x}}
\def\by{\mathbf{y}}
\def\bz{\mathbf{z}}
\def\bw{\mathbf{w}}

\newcommand{\EDB}{\mathit{ED}}

\newcommand{\extdb}{extensional database\xspace}

\newcommand{\program}{program\xspace}

\newcommand{\PP}{\P_{\EDB}}

\newtheorem{theorem}{Theorem}
\newtheorem{lemma}[theorem]{Lemma}
\newtheorem{example}{Example}
\newtheorem{definition}{Definition}
\newtheorem{corollary}{Corollary}

\title[Theory and Practice of Logic Programming]
{Querying Data Exchange Settings Beyond Positive Queries\footnote{Competing interests: The author(s) declare none}}

\author[Marco Calautti, Sergio Greco, Cristian Molinaro and Irina Trubitsyna]{
	MARCO CALAUTTI$^1$, SERGIO GRECO$^2$,  \\
	{\normalsize \em CRISTIAN MOLINARO$^2$ and IRINA TRUBITSYNA$^2$}\\
	$^1$DI, University of Milan, Italy\\
	$^2$DIMES, University of Calabria, Italy\\
	\email{marco.calautti@unimi.it, \{greco,cmolinaro,trubitsyna\}@dimes.unical.it}
}

\jdate{March 2003}
\pubyear{2003}
\pagerange{\pageref{firstpage}--\pageref{lastpage}}
\doi{S1471068401001193}

\begin{document}
	
	\label{firstpage}
	
	\maketitle
	
	\begin{abstract}
	Data exchange, the problem of transferring data from a source schema to a target schema, has been studied for several years. 
	The semantics of answering positive queries over the target schema has been defined in early work, but little attention has been paid to more general queries. A few proposals of semantics for more general queries exist but they either do not properly extend the standard semantics under positive queries, giving rise to counterintuitive answers, or they make query answering undecidable even for the most important data exchange settings, e.g., with weakly-acyclic dependencies.
	
	The goal of this paper is to provide a new semantics for data exchange that is able to deal with general queries. At the same time, we want our semantics to coincide with the classical one when focusing on positive queries, and to not trade-off too much in terms of complexity of query answering. We show that query answering is undecidable in general under the new semantics, but it is $\co\NP\complete$ when the dependencies are weakly-acyclic.
	Moreover, in the latter case, we show that exact answers under our semantics can be computed by means of logic programs with choice, thus exploiting existing efficient systems. For more efficient computations, we also show that our semantics allows for the construction of a representative target instance, similar in spirit to a universal solution, that can be exploited for computing approximate answers in polynomial time.
	%
	\end{abstract}
	
	\begin{keywords}
		Data Exchange, Semantics, Closed Word Assumption, Approximations
	\end{keywords}
	
	\tableofcontents

\section{Introduction}
    Data exchange is the problem of transferring data from a source schema to a target schema, where the transfer process is usually described via so-called schema mappings: a set of logical assertions specifying how the data should be moved and restructured. 
    Furthermore, the target schema may have its own constraints to be satisfied. 
    Schema mappings and target constraints are usually encoded via standard database dependencies: \emph{tuple-generating dependencies} (TGDs) and \emph{equality-generating dependencies} (EGDs).
	Thus, given an instance $I$ over the source schema $\Src$, the goal is to materialize an instance $J$ over the target schema $\Trg$, called \emph{solution}, in such a way that $I$ and $J$ together satisfy the dependencies.
	
	Since multiple solutions might exist, a precise semantics for answering queries is needed. By now, the \emph{certain answers} semantics is the most accepted one. The certain answers to a query is the set of all tuples that are answers to the query in every solution of the data exchange setting~\cite{FaginKMP05}. Although it has been formally shown that for positive queries (e.g., conjunctive queries) the notion of solution of~\cite{FaginKMP05} is the right one to use, for more general queries such solutions become inappropriate, as they easily lead to counterintuitive  results. 
  	 
  	\begin{example}\label{ex:wrong-answer}
  		Consider a data exchange setting denoted by $\Set = \des$, where $\Src$ is the source schema, storing product orders  in a binary relation $\Ord$, with the first argument being the id of an order, and the second argument specifying whether the order has been paid. Moreover, $\Trg$ is the target schema having unary relations $\AllOrd$ and $\Paid$, storing all orders and the paid orders, respectively. The schema mapping is described by the following source-to-target TGDs $\STTGD$:
  		\[
  		\begin{array}{llllll}
  			\rho_1 = &  \forall x,y & \Ord(x,y) \rightarrow \AllOrd(x),\qquad \rho_2 = &  \forall x & \Ord(x,\mathsf{yes}) \rightarrow \Paid(x). 
  		\end{array}
  		\]
  		In this example, we assume that the set of target dependencies $\TTGD$ is empty.
  		The above schema mapping states that all orders in the source schema must be copied to the $\AllOrd$ relation, and all the paid orders must be copied to the $\Paid$ relation.
  		Assume the source instance is as follows:
  		$$I=\{\Ord(1,\mathsf{yes}),\Ord(2,\mathsf{no})\},$$
  		and assume we want to pose the query $Q$ over the target schema asking for all the unpaid orders. This can be written as the following FO query:
  		$$ Q(x) = \AllOrd(x) \wedge \neg \Paid(x).$$
  		One would expect the answer to be $\{2\}$, since the schema mapping above is simply copying $I$ to the target schema, and hence $J = \{\AllOrd(1),\AllOrd(2),\Paid(1)\}$ should be the only candidate solution. However, under the classical notion of solution of~\cite{FaginKMP05},  also the instance $J' = \{\AllOrd(1),\AllOrd(2),\Paid(1),\Paid(2)\}$ is a solution (since $I \cup J'$ satisfies the TGDs), and every order in $J'$ is paid. Hence, the certain answers to $Q$, which are computed as the intersection of the answers over all solutions, are empty.
  	\end{example}
  	The issue above arises because the classical notion of solution is too permissive, in that it allows the existence of facts in a solution that have no support from the source (e.g., $\Paid(2)$ in the solution $J'$ of Example~\ref{ex:wrong-answer} above).
  	
  	Some efforts exist in the literature that provide alternative notions of solutions for which certain answers to general queries become more meaningful. Prime examples are the works of~\cite{HernichLibkin2011} and~\cite{Hernich2011}. In both approaches,
  	the certain answers in the example above are $\{2\}$.
  	However, the works above have their own drawbacks too. In~\cite{HernichLibkin2011}, so-called \emph{CWA-solutions} are introduced, which are a subset of the classical solutions with some restrictions. 
  	However, these restrictions are so severe that certain answers over such solutions fail to capture certain answers over classical solutions, when focusing on positive queries. 
  	Moreover, even when focusing on more general queries, answers can still be counterintuitive, as shown in the following example.
  	
  	\begin{example}\label{ex:employee}
  		Consider the data exchange setting $\Set = \des$, where $\Src$ stores employees of a company in the unary relation $\Emp$. For \emph{some} employees, the city they live in is known, and it is stored in the binary relation $\EmpC$. 
  		The target schema $\Trg$ contains the binary relation $\EmpNew$, storing employees and the cities they live in, and the binary relation $\SameC$, storing pairs of employees living in the same city.
  		The sets $\STTGD = \{\rho_1,\rho_2\}$ and $\TTGD = \{\rho_3,\eta\}$ are as follows (for simplicity, we omit the universal quantifiers):
  		\[
  		\begin{array}{l}
  		\rho_1 = \Emp(x) \rightarrow \exists z\, \EmpNew(x,z), \\
  		\rho_2 = \EmpC(x,y) \rightarrow  \EmpNew(x,y), \\
		\rho_3 = \EmpNew(x,y),\ \EmpNew(x',y) \rightarrow \SameC(x,x'), \\
  		\eta\  = \EmpNew(x,y),\ \EmpNew(x,z) \rightarrow y=z.\\
  		\end{array}
  		\]
  		The above setting copies employees from the source to the target. The TGD $\rho_1$ states that every copied employee $x$ must have some city $z$ associated, whereas $\rho_2$ states that when the city $y$ of an employee $x$ is known, this should be copied as well.
  		Moreover, the target schema requires that employees living in the same city should be stored in relation $\SameC$ ($\rho_3$), and each employee must live in only one city ($\eta$).
  		Assume the source instance is 
  		$$I = \{\Emp(\rel{john}), \Emp(\rel{mary}),\EmpC(\rel{\john},\rel{miami})\},$$
  		and consider the query $Q$ that asks for all pairs of employees living in different cities. This can be written as:
$$
  		Q(x,x') = \exists y \exists y'\, \EmpNew(x,y) \wedge \EmpNew(x',y') 
  		\wedge \neg \SameC(x,x').
$$
  		One would expect that the set of certain answers to $Q$ is empty, since it is not certain that $\rel{john}$ and $\rel{mary}$ live in different cities. However, no CWA-solution admits $\rel{mary}$ and $\rel{john}$ to live in the same city, and thus $(\rel{\john},\rel{mary})$ is a certain answer under the CWA-solution-based semantics.
  	\end{example}
  
  	The approach of~\cite{Hernich2011}, where the notion of GCWA$^*$-solution is presented, seems to be the most promising one. 
For positive queries, certain answers w.r.t.\ GCWA$^*$-solutions coincide with certain answers w.r.t.\ classical solutions. Moreover, GCWA$^*$-solutions solve some other limitations of CWA-solutions, like the one discussed in Example~\ref{ex:employee}. However, the practical applicability of this semantics is somehow limited, since the (rather involved) construction of GCWA$^*$-solutions easily makes certain query answering undecidable, even for very simple settings with only two source-to-target TGDs, and no target dependencies.
  	
  	Other semantics have been proposed in~\cite{LibkinSirangelo2011}, but they are only defined for data exchange settings without target dependencies. Hence, one needs to assume that the target schema has no dependencies at all.
  	
  	As a final remark, in a data exchange setting, it might be the case that the source is not always available, and thus the materialization of a single solution, over which certain answers can be computed, is a desirable requirement. This is especially true when using weakly-acyclic dependencies, which form the standard language for data exchange~\cite{FaginKMP05}. However, none of the semantics above allow for the materialization of such a special solution, for weakly-acyclic settings.
  	
 	\smallskip
	In this paper, we propose a new notion of data exchange solution, dubbed \emph{supported solution}, which allows us to deal with general queries, but at the same time is suitable for practical applications.
  	That is, we show that certain answers under supported solutions naturally generalize certain answers under classical solutions, when focusing on positive queries. Moreover, such solutions do not make any assumption on how values associated to existential variables compare to other values, hence solving issues like the ones of Example~\ref{ex:employee}.
  	
  	As expected, there is a price to pay to get meaningful answers over general queries: we show that certain answering is undecidable for general settings, but it becomes $\co\NP\complete$ when we focus on weakly-acyclic dependencies.
  	
  	Moreover, we show that exact answers under supported solutions for \emph{general} queries in weakly-acyclic settings can be computed via an encoding into logic programming with the well-known choice construct, allowing one to use efficient off-the-shelf reasoning systems.
  	
  	Finally, we also show that if one is not willing to incur the high complexity of exact certain answers for weakly-acyclic settings, then it is possible to construct a target instance in \emph{polynomial time}, which is similar in spirit to a universal solution of~\cite{FaginKMP05}, that can be exploited for computing exact answers, for positive queries, and \emph{approximate answers}, for general FO queries, in polynomial time. The latter is achieved by adapting existing approximation algorithms originally defined for querying incomplete databases.
  	%
  	
\section{Preliminaries}\label{sec:preliminaries}


\textbf{Basics.}
We consider pairwise disjoint countably infinite sets $\Const$, $\Var$, and $\Null$ of \emph{constants}, \emph{variables}, and \emph{labeled nulls}, respectively.
Nulls are denoted by the symbol $\nul$, possibly subscripted. 
A \emph{term} is a constant, a variable, or a null.
We additionally assume the existence of a countably infinite set $\Rel$ of \emph{relations}, disjoint from the previous ones.
A relation $R$ has an \emph{arity}, denoted $\ar(R)$, which is a non-negative integer.  
We also use $R/n$ to say that $R$ is a relation of arity $n$.
A \emph{schema} is a set of relations. A \emph{position} is an expression of the form $R[i]$, where $R$ is a relation and $i \in \{1,\ldots,\ar(R)\}$.

An \emph{atom} $\alpha$ (over a schema $\Sch$) is of the form $R(\bt)$, where $R$ is an $n$-ary relation (of $\Sch$) and $\bt$ is a tuple of terms of length $n$.
We use $\bt[i]$ to denote the $i$-th term in $\bt$, for $i \in \{1,\ldots,n\}$.
An atom without variables is a \emph{fact}.
An \emph{instance} $I$ (over a schema $\Sch$) is a finite set of facts (over $\Sch$).
A \emph{database} $D$ is an instance without nulls.
For a set of atoms $A$, $\adom(A)$ is the set of all terms in $A$, whereas $\var(A)$ is the set $\adom(A) \cap \Var$.
A \emph{homomorphism} from a set of atoms $A$ to a set of atoms $B$ is a function $h : \adom(A) \rightarrow \adom(B)$ that is the identity on $\Const$, and such that for each atom $R(\bt) = R(t_1,\ldots,t_n) \in A$, $R(h(\bt))=R(h(t_1),\ldots,h(t_n)) \in B$.

\smallskip
\noindent
\textbf{Dependencies.}
A \emph{tuple-generating dependency} (TGD) $\rho$ (over a schema $\Sch$) is a first-order formula of the form $\forall \x,\y\, \varphi(\x,\y) \rightarrow \exists \z\, \psi(\y,\z)$, where $\x,\y,\z$ are disjoint tuples of variables, and $\varphi$ and $\psi$ are conjunctions of atoms  (over $\Sch$) without nulls, and over the variables in $\x,\y$ and $\y,\z$, respectively. The \emph{body} of $\rho$, denoted $\body(\rho)$, is $\varphi(\x,\y)$, whereas the \emph{head} of $\rho$, denoted $\head(\rho)$, is $\psi(\y,\z)$. We use $\exvar(\rho)$ to denote the tuple $\z$ and $\fr(\rho)$ to denote the tuple $\y$, also called the \emph{frontier of $\rho$}.
An \emph{equality-generating dependency} (EGD) $\eta$ (over a schema $\Sch$) is a first-order formula of the form $\forall \x\, \varphi(\x) \rightarrow x=y$, where $\x$ is a tuple of variables, $\varphi$ a conjunction of atoms (over $\Sch$) without nulls, and over $\x$, and $x,y \in \x$. The \emph{body of $\eta$}, denoted $\body(\eta)$, is $\varphi(\x)$, and the \emph{head of $\eta$}, denoted $\head(\eta)$, is the equality $x=y$.
For clarity, we will omit the universal quantifiers in front of dependencies and replace the conjunction symbol $\wedge$ with a comma. Moreover, with a slight abuse of notation, we sometimes treat a conjunction of atoms as the \emph{set} of its atoms.
%
Consider an instance $I$. We say that $I$ \emph{satisfies a TGD} $\rho$ if for every homomorphism $h$ from $\body(\rho)$ to $I$, there is an extension $h'$ of $h$ such that $h'$ is a homomorphism from $\head(\rho)$ to $I$. We say that $I$ \emph{satisfies an EGD} $\eta = \varphi(\x) \rightarrow x=y$, if for every homomorphism $h$ from $\body(\eta)$ to $I$, $h(x)=h(y)$.
$I$ \emph{satisfies} a set of TGDs and EGDs $\Sigma$ if $I$ satisfies every TGD and EGD in $\Sigma$.

\smallskip
\noindent
\textbf{Queries.} A \emph{query} $Q(\bx)$, with free variables $\bx$, is a first-order (FO) formula $\varphi(\x)$ with free variables $\x$. The \emph{arity} of $Q(\bx)$, denoted $\ar(Q)$, is the number $|\bx|$. The \emph{output} of $Q(\bx)$ over an instance $I$, denoted $Q(I)$, is the set $\{\bt \in \adom(I)^{|\bx|} \mid I \models \varphi(\bt)\}$, where $\models$ is FO entailment.\footnote{We assume active domain semantics, i.e., quantifiers range over the terms in the given instance.}
A query is \emph{Boolean} if it has arity 0, in which case its output over an instance is either the empty set or the empty tuple $\langle \rangle$.
A \emph{conjunctive query} (CQ) is a query of the form $Q(\bx) = \exists \by\, \varphi(\bx,\by)$, where $\varphi(\bx,\by)$ is a conjunction of atoms over $\bx$ and $\by$. A \emph{union of conjunctive queries} (UCQ) is a query of the form $Q(\bx) = \bigvee^n_{i=1} Q_i(\bx)$, where each  $Q_i(\x)$ is a CQ. We refer to UCQs also as \emph{positive queries}.

\smallskip
\noindent
\textbf{Data Exchange Settings.}
A \emph{data exchange setting} (or simply setting) is a tuple of the form $\Set = \langle \Src, \Trg, \STTGD, \TTGD \rangle$, where $\Src,\Trg$ are disjoint schemas, called \emph{source} and \emph{target} schema, respectively; $\STTGD$ is a finite set of TGDs, called the \emph{source-to-target TGDs} of $\Set$, such that for each TGD $\rho \in \STTGD$, $\body(\rho)$ is over $\Src$ and $\head(\rho)$ is over $\Trg$; $\TTGD$ is a finite set of TGDs and EGDs over $\Trg$, called the \emph{target dependencies} of $\Set$. We say $\Set$ is \emph{TGD-only} if $\TTGD$ contains only TGDs.

A \emph{source (resp., target) instance of} $\Set$ is an instance $I$ over $\Src$ (resp., $\Trg$). We assume that source instances are databases, i.e., they do not contain nulls.
Given a source instance $I$ of $\Set$, a \emph{solution of $I$ w.r.t.\ $\Set$} is a target instance $J$ of $\Set$ such that $I \cup J$ satisfies $\STTGD$ and $J$ satisfies $\TTGD$~\cite{FaginKMP05}. We use $\sol{I,\Set}$ to denote the set of all solutions of $I$ w.r.t.\ $\Set$.

Given a data exchange setting $\Set = \langle \Src, \Trg, \STTGD, \TTGD \rangle$, a source instance $I$ of $\Set$ and a query $Q$ over $\Trg$,  the \emph{certain answers to $Q$ over $I$ w.r.t.\ $\Set$} is the set
$\cert{I}{Q}{\Set} = \bigcap_{J \in \sol{I,\Set}} Q(J)$.

To distinguish between the notion of solution (resp., certain answers) above and the one defined in Section~\ref{sec:semantics}, we will refer to the former  as \emph{classical}.

A \emph{universal solution} of $I$ w.r.t.\ $\Set$ is a solution $J \in \sol{I,\Set}$ such that, for every $J' \in \sol{I,\Set}$, there is a homomorphism from $J$ to $J'$~\cite{FaginKMP05}. Letting $\nonull{Q(J)} = Q(J) \cap \Const^{|\x|}$, for any instance $J$ and query $Q(\x)$, the following result from~\cite{FaginKMP05} is well-known:

\begin{theorem}\label{thm:universal-fagin}
	Consider a data exchange setting $\Set$, a source instance $I$ of $\Set$ and a positive query $Q$. If $J$ is a universal solution of $I$ w.r.t.\ $\Set$, then
	$ \cert{I}{Q}{\Set} = \nonull{Q(J)}$.
\end{theorem}

\section{Semantics for General Queries}\label{sec:semantics}
The goal of this section is to introduce a new notion of solution for data exchange that we call \emph{supported}. As already discussed, the main issue we want to solve w.r.t.\ classical solutions is that such solutions are too permissive, i.e., they allow for the presence of facts that are not a certain consequence of the source instance and the dependencies.
Consider again Example~\ref{ex:wrong-answer}. The (classical) solution $J'$ in Example~\ref{ex:wrong-answer} is not supported, since from the source instance $I$ and the dependencies, we cannot conclude that the fact $\Paid(2)$ should occur in the target. On the other hand, the solution $J = \{\AllOrd(1),\AllOrd(2),\Paid(1)\}$ is supported: it contains precisely the facts supported by $I$ and the dependencies, and no more than that.
Similarly, considering Example~\ref{ex:employee}, the instance $J = \{\EmpNew(\rel{john},\rel{miami})$, $\EmpNew(\rel{mary},\rel{chicago})$, $\SameC(\rel{john},\rel{mary})\}$ is a solution, but it is not supported, since from the source and the dependencies we cannot certainly conclude that $\rel{john}$ and $\rel{mary}$ live in the same city. We now formalize the above intuitions.

Consider a TGD $\rho$ and a mapping $h$ from the variables of $\rho$ to $\Const$. We say that a TGD $\rho'$ is a \emph{ground version of $\rho$} (via $h$) if $\rho' = h(\body(\rho)) \rightarrow h(\head(\rho))$.

\begin{definition}[\exchoice]\label{def:exchoice}
	An \emph{\exchoice} is a function $\gamma$, that given as input a TGD $\rho = \varphi(\x,\y) \rightarrow \exists \z\, \psi(\y,\z)$ and a tuple $\bt \in \Const^{|\y|}$, returns a set $\gamma(\rho,\bt)$ of pairs  of the form $(z,c)$, one for each existential variable $z \in \exvar(\rho)$, where $c$ is a constant of $\Const$.
\end{definition}
 Note that if $\rho$ does not contain existential variables, $\gamma(\rho,\bt)$ is the empty set.

Intuitively, given a TGD, an \exchoice\ specifies a valuation for the existential variables of the TGD which depends on a given valuation of its frontier variables.

\smallskip
We now define when a ground version of a TGD indeed assigns existential variables according to an \exchoice .

\begin{definition}[Coherence]
	Consider a TGD $\rho = \varphi(\x,\y) \rightarrow \exists \z\, \psi(\y,\z)$, an \exchoice\ $\gamma$ and a ground version $\rho'$ of $\rho$ via some mapping $h$. We say that $\rho'$ is \emph{coherent with $\gamma$} if for each existential variable $z \in \exvar(\rho)$, $(z,h(z)) \in \gamma(\rho,h(\y))$. 
\end{definition}
For a set $\Sigma$ of TGDs and EGDs, and an \exchoice\ $\gamma$, $\Sigma^\gamma$ denotes the set of dependencies obtained from $\Sigma$ by replacing each TGD $\rho$ in $\Sigma$ with all ground versions of $\rho$ that are coherent with $\gamma$. Note that the set $\Sigma^\gamma$ can be infinite.
We are now ready to present our notion of solution.

\begin{definition}[Supported Solution]\label{def:supported-solution}
Consider a setting $\Set = \des$ and a source instance $I$ of $\Set$. A target instance $J$ of $\Set$ is a \emph{supported solution of $I$ w.r.t.\ $\Set$} if there exists an \exchoice\ $\gamma$ such that $I \cup J$  satisfies $\STTGD^\gamma$ and $J$ satisfies $\TTGD^\gamma$, and there is no other target instance $J' \subsetneq J$ of $\Set$ such that $I \cup J'$ satisfies $\STTGD^\gamma$ and $J'$ satisfies $\TTGD^\gamma$.
\end{definition}

Note that a supported solution contains no nulls.
We use $\ssol{I,\Set}$ to denote the set of all supported solutions of $I$ w.r.t.\ $\Set$.

\begin{example}
	Consider the data exchange setting $\Set$ and the source instance $I$ of Example~\ref{ex:employee}.
	The target instance 
	$J = \{\EmpNew(\rel{john},\rel{miami}),\EmpNew(\rel{mary},\rel{chicago})\}$
	is a supported solution of $I$ w.r.t.\ $\Set$. Indeed, consider the \exchoice\ $\gamma$ such that $\gamma(\rho_1,\rel{john}) = \{ (z,\rel{miami})\}$, and $\gamma(\rho_1,\rel{mary}) = \{(z,\rel{chicago})\}$. Then,  $\STTGD^\gamma$ is
	\[
	\begin{array}{ll}
	\{\EmpC(\alpha,\beta) \rightarrow \EmpNew(\alpha,\beta) \mid \alpha,\beta \in \Const\} \cup \\
	 \{\Emp(\alpha) \rightarrow \EmpNew(\alpha,\beta) \mid  \alpha \in \Const \wedge (z,\beta) \in \gamma(\rho_1,\alpha)\},
	\end{array}
	\]
	whereas $\TTGD^\gamma$ is the set containing the EGD $\eta$ of Example~\ref{ex:employee}, and the set of TGDs
	\[
	\begin{array}{l}
	 \{\EmpNew(\alpha,\beta),\EmpNew(\alpha',\beta) \rightarrow \SameC(\alpha,\alpha') \mid 
	\alpha,\alpha',\beta \in \Const\}.
	\end{array}
	\]
	Clearly, $I \cup J$ satisfies $\STTGD^\gamma$, and $J$ satisfies $\TTGD^\gamma$, and any other strict subset $J'$ of $J$ is such that $I \cup J'$ does not satisfy $\STTGD^\gamma$. Another supported solution is
$\{\EmpNew(\rel{john},\rel{miami})$, $\EmpNew(\rel{mary},\rel{miami})$, $\SameC(\rel{john},\rel{mary})\}$.
	\end{example}

With the notion of supported solution in place, it is now straightforward to define the supported certain answers.

\begin{definition}[Supported Certain Answers]
    Consider a data exchange setting $\Set$, a source instance $I$ of $\Set$ and a query $Q$ over $\Trg$. The \emph{supported certain answers to $Q$ over $I$ w.r.t.\ $\Set$} is the set of tuples $\scert{I}{Q}{\Set} = \bigcap_{J \in \ssol{I,\Set}} Q(J)$.

\end{definition}

\begin{example}
	Consider the data exchange setting $\Set$, the source instance $I$, and the query $Q$ of Example~\ref{ex:wrong-answer}. It is not difficult to see that the only supported solution of $I$ w.r.t.\ $\Set$ is the instance
	$$J = \{\AllOrd(1),\AllOrd(2), \Paid(1)\}.$$
	Thus, the supported certain answers to $Q$ over $I$ w.r.t.\ $\Set$ are
	$ \scert{I}{Q}{\Set} = Q(J) = \{2\}$.
	Consider now the data exchange setting $\Set$, the source instance $I$, and the query $Q$ of Example~\ref{ex:employee}. Then, one can verify that
	$\scert{I}{Q}{\Set} = \emptyset$.
\end{example}

We now start establishing some important results regarding supported solutions and supported certain answers. The following theorem states that supported solutions are a refined subset of the classical ones, but whether a supported solution exists is still  tightly related to the existence of a classical one.

\begin{theorem}\label{thm:sol-relationship}
	Consider a data exchange setting $\Set$. For every source instance $I$ of $\Set$, it holds that:
	\begin{enumerate}[align=parleft,left=0pt..1em]
		\item $\ssol{I,\Set} \subseteq \sol{I,\Set}$, and 
		\item $\ssol{I,\Set} = \emptyset$ iff $\sol{I,\Set} = \emptyset$.
	\end{enumerate}
\end{theorem}
\begin{proof}
	Item 1 follows by definition. For proving Item 2, it suffices to show that $\sol{I,\Set} \neq \emptyset$ implies $\ssol{I,\Set} \neq \emptyset$.
	Let $\Set = \des$ and consider a solution $J \in \sol{I,\Set}$. We construct from $J$ a supported solution $\hat{J}$ in $\ssol{I,\Set}$. Let
	$J'$ be one of the minimal subsets of $J$ such that $J'$ is still a solution of $\sol{I,\Set}$. Moreover, let $\hat{J}$ be the instance obtained from $J'$, where each null $\perp$ in $J'$ is replaced with a new constant $c_\perp$ not occurring in $\STTGD \cup \TTGD$ and $J'$. Since $\hat{J}$ and $J'$ are the same instance, up to null renaming, we conclude that $\hat{J}$ is also a solution in $\sol{I,\Set}$. To see that $\hat{J}$ is a supported solution,  consider the following \exchoice\ $\gamma$. For every TGD $\rho \in \STTGD \cup \TTGD$, and every tuple $\bt$ of constants such that there exists a homomorphism $h$ from $\body(\rho)$ to $\hat{J}$, and $\bt = h(\fr(\rho))$, let $\gamma(\rho,\bt) = \{(z,h(z)) \mid z \in \exvar(\rho)\}$. By construction of $\gamma$, $I \cup \hat{J}$ satisfies $\STTGD^\gamma$, and $\hat{J}$ satisfies $\TTGD^{\gamma}$. Since $\hat{J}$ is minimal, i.e., for every $J'' \subsetneq \hat{J}$, $J'' \not \in \sol{I,\Set}$, from Item 1 of this claim, every $J'' \subsetneq J$ is such that $J'' \not \in \ssol{I,\Set}$, i.e., either $I \cup J''$ does not satisfy $\STTGD^\gamma$ or $J''$ does not satisfy $\TTGD^\gamma$. Thus, $\hat{J}$ is a supported solution of $\ssol{I,\Set}$, and the claim follows.
\end{proof}

Regarding certain answers, we show that supported solutions indeed enjoy an important property: supported certain answers and classical certain answers coincide, when focusing on positive queries. Note that this does not necessarily follow from Theorem~\ref{thm:sol-relationship}.

\begin{theorem}\label{thm:equiv-positive}
    Consider a setting $\Set = \des$ and a positive query $Q$ over $\Trg$.
    For every source instance $I$ of $\Set$,
    $\scert{I}{Q}{\Set} = \cert{I}{Q}{\Set}$.
\end{theorem}
\begin{proof}
	The fact that $\cert{I}{Q}{\Set} \subseteq \scert{I}{Q}{\Set}$, follows from Item 1 of Theorem~\ref{thm:sol-relationship}. To prove that $\scert{I}{Q}{\Set} \subseteq \cert{I}{Q}{\Set}$, assume $\bt \not \in \cert{I}{Q}{\Set}$, which means that there exists a solution $J$ of $I$ w.r.t.\ $\Set$ such that $\bt \not \in Q(J)$.
	Since $Q$ is positive, and hence monotone, $\bt \not \in Q(J)$ iff $\bt \not \in Q(J')$, where $J'$ is one of the minimal subsets of $J$ such that $J'$ is still a solution of $I$ w.r.t.\ $\Set$.
	Let $\hat{J}$ be the instance obtained from $J'$, where each null $\perp$ in $J'$ is replaced with a new constant $c_\perp$ not occurring in $\bt$, $Q$, $\STTGD \cup \TTGD$, and $J'$. With a similar discussion to the one given in the proof of Theorem~\ref{thm:sol-relationship}, we conclude that $\hat{J}$ is a supported solution of $I$ w.r.t.\ $\Set$. Since $Q$ is positive, and since $\bt$ and $Q$ do not contain any of the constants introduced in $J'$, we conclude that $\bt \not \in Q(\hat{J})$, which implies that $\bt \not \in \scert{I}{Q}{\Set}$, and the claim follows.
\end{proof}
From the above, we conclude that for positive queries, certain query answering can be performed as done in the classical setting, and thus all important results from that setting, like query answering via universal solutions, carry over.

\begin{corollary}
	Consider a setting $\Set = \des$ and a positive query $Q$ over $\Trg$. If $J$ is a (classical) universal solution of $I$ w.r.t.\ $\Set$, then
	$ \scert{I}{Q}{\Set} = \nonull{Q(J)}$.
\end{corollary}
\begin{proof}
	It follows from Theorem~\ref{thm:universal-fagin} and Theorem~\ref{thm:equiv-positive}.
\end{proof}

We now move to the complexity analysis of the two most important data exchange tasks:
deciding whether a supported solution exists, and computing the supported certain answers to a query.

\section{Complexity}
In data exchange, it is usually assumed that a setting $\Set$ does not change over time, and a given query $Q$ is much smaller than a given source instance. Thus, for understanding the complexity of a data exchange problem, it is customary to assume that $\Set$ and $Q$ are fixed, and only $I$ is considered in the complexity analysis, i.e., we consider the \emph{data complexity} of the problem.
Hence, the problems we are going to discuss will always be parametrized via a setting $\Set$, and a query $Q$ (for query answering tasks).
The first problem we consider is deciding whether a supported solution exists; $\Set$ is a fixed data exchange setting.
\begin{center}
	\framebox[8cm]{
			\begin{tabular}{ll}
			{\small PROBLEM} :  & $\supexprob{\Set}$
			\\
			{\small INPUT} :    & A source instance $I$ of $\Set$.
			\\
			{\small QUESTION} : & Is $\ssol{I,\Set} \neq \emptyset$?
			\end{tabular}}
\end{center}
The above problem is very important in data exchange, as one of the main goals is to actually construct a target instance that can be exploited for query answering purposes. Hence, knowing in advance whether at least a supported solution exists is of paramount importance.

Thanks to Item~2 of Theorem~\ref{thm:sol-relationship}, all the complexity results for checking the existence of a classical solution can be directly transferred to our problem.


\begin{theorem}
	There exists a data exchange setting $\Set$ such that $\supexprob{\Set}$ is undecidable.
\end{theorem}
\begin{proof}
	It follows from Theorem~\ref{thm:sol-relationship} and from the fact that there exists a data exchange setting $\Set$ such that checking whether a classical solution exists is undecidable~\cite{KolaitisPT06}.
\end{proof}

Despite the negative result above, we also inherit positive results from the literature, when focusing on some of the most important data exchange scenarios, known as \emph{weakly-acyclic}. Such settings only allow target TGDs to belong to the language of weakly-acyclic TGDs, which have been first introduced in the seminal paper~\cite{FaginKMP05}, and is now well-established as the main language for data exchange purposes.

We start by introducing the notion of weak-acyclicity. We recall that for a schema $\Sch$, $\pos{\Sch}$ denotes the set of all positions $R[i]$, where $R/n \in \Sch$ and $i \in \{1,\ldots,n\}$, and for a TGD $\rho = \varphi(\x,\y) \rightarrow \exists \z\, \psi(\y,\z)$, $\fr(\rho)$ denotes the tuple $\y$.

\begin{definition}[Dependency Graph~\cite{FaginKMP05}]\label{def:dependency-graph}
	Consider a set $\Sigma$ of TGDs over a schema $\Sch$. The \emph{dependency graph} of $\Sigma$ is a directed graph $\depg{\Sigma}=(N,E)$, where $N = \pos{\Sch}$ and $E$ contains only the following edges.
	For each $\rho \in \Sigma$, for each $x \in \fr(\rho)$, and
	for each position $\pi$ in $\body(\rho)$ where $x$ occurs:
	
	\begin{itemize}[align=parleft,left=0pt..1em]
		\item there is a \emph{normal} edge $(\pi,\pi') \in E$, for each position $\pi'$ in $\head(\rho)$ where $x$ occurs, and
		\item there is a \emph{special} edge $(\pi,\pi') \in E$, for each position $\pi'$ in $\head(\rho)$ where an existentially quantified variable $z \in \exvar(\rho)$ occurs.
	\end{itemize}
\end{definition}

\begin{definition}\label{def:wa}
	A set of TGDs $\Sigma$ is \emph{weakly-acyclic} if no cycle in 
	$\depg{\Sigma}$ contains a special edge. A data exchange setting $\des$ is \emph{weakly-acyclic} if the set of TGDs in $\TTGD$ is weakly-acyclic.
\end{definition}

\begin{example}
	The settings of Examples~\ref{ex:wrong-answer} and~\ref{ex:employee} are weakly-acyclic, whereas the data exchange setting $\Set = \des$, where $\Src = \{S/2\}$, $\Trg = \{T/2\}$, $\STTGD = \{S(x,y) \rightarrow T(x,y)\}$, and $\TTGD = \{T(x,y) \rightarrow \exists z\, T(y,z)\}$ is not, since $(T[2],T[2])$ is a special edge in $\depg{\TTGD}$.
\end{example}

The following result follows.
\begin{theorem}\label{thm:exsol-wa}
	For every weakly-acyclic data exchange setting $\Set$, $\supexprob{\Set}$ is in $\PTIME$.
\end{theorem}
\begin{proof}
	It follows from Theorem~\ref{thm:sol-relationship} and \cite[Corollary~3.10]{FaginKMP05}.
\end{proof}

We now move to the second crucial task: computing supported certain answers.
Since this problem outputs a set, it is standard to focus on its decision version.
For a fixed data exchange setting $\Set$ and a fixed query $Q$, we consider the following decision problem:
\begin{center}
	
	\framebox[12cm]{\begin{tabular}{ll}
			{\small PROBLEM} :  & $\scertprob{\Set,Q}$
			\\
			{\small INPUT} :    & A source instance $I$ of $\Set$ and 
			 a tuple $\bt \in \Const^{\ar(Q)}$.
			\\
			{\small QUESTION} : & Is $\bt \in \scert{I}{Q}{\Set}$?
	\end{tabular}}
\end{center}

One can easily show that the above problem is logspace equivalent to the one of computing the supported certain answers. 

%
%
%
%

We start by studying the problem in its full generality, and show that there is a price to pay for query answering with general queries.

\begin{theorem}\label{thm:scert-undecidable}
	There exists a data exchange setting $\Set = \des$, with $\TTGD$ having only TGDs, and a query $Q$ over $\Trg$, such that $\scertprob{\Set,Q}$ is undecidable. 
\end{theorem}
\begin{proof}
	We provide a polynomial-time reduction from the \emph{Embedding Problem for Finite Semigroups $\mathsf{EMB}$~\cite{KolaitisPT06}}. The reduction is an adaptation of the one used for proving Proposition~6.1 in \cite{HernichLibkin2011}.
	Inputs of $\mathsf{EMB}$ are pairs of the form $A,f$, where $A$ is a finite set, and $f$ is a partial function of the form $f: A \times A \rightarrow A$. The question is whether there exists a finite set $B \supseteq A$, and a total function $g : B \times B \rightarrow B$, such that $g$ is associative\footnote{A total function $g: B \times B \rightarrow B$ is associative if for every $a,b,c \in B$, $g(g(a,b),c) = g(a,g(b,c))$.}, and $g$ extends $f$, i.e., whenever $f(a,b)$ is defined, $g(a,b) = f(a,b)$.
	
	%
	Let us first introduce some notation. Consider a finite set $A$ and a partial function $f : A \times A \rightarrow A$. We define the instance:
	$$ I_{A,f} = \{\rel{F}(a,b,c) \mid a,b,c \in A \text{ and } f(a,b) = c\}.$$
	
	Consider now the data exchange setting $\Set = \des$, where $\Src = \{\rel{F}/3\}$ and $\Trg = \{\rel{G}/3\}$. Intuitively, the relation $\rel{F}$ collects all the triples $a,b,c$ such that $f(a,b) = c$, whereas the relation $\rel{G}$ collects all the triples of the extended associative function $g$. The sets $\STTGD$ and $\TTGD$ are defined as $\STTGD = \{\rel{F}(x,y,z) \rightarrow \rel{G}(x,y,z)\}$ and $\TTGD = \{\rel{G}(x,y,z) \rightarrow \exists x',y',z'\ \rel{G}(x',y',z') \wedge \rel{Aux}(x,y,z)\}$. Roughly $\STTGD$ is in charge of forcing the function stored in $\rel{G}$ to be an extension of the function stored in $\rel{F}$, whereas $\TTGD$ is in charge of adding additional entries to $\rel{G}$.
	
	The difference with the construction of \cite{HernichLibkin2011} is in the set $\TTGD$. Here, the head of the only TGD in $\TTGD$ has an additional auxiliary atom $\rel{Aux}(x,y,z)$. 
	Intuitively, since the set $\TTGD$ is in charge of extending the function defined by the relation $\rel{F}$ by introducing additional terms, in order for these terms to be actually introduced in a supported solution, we require that every body variable is also a frontier variable. 
	Regarding our query $Q$, it is the same as the one in~\cite{HernichLibkin2011}.
	Hence, instead of giving the precise expression of $Q$, we only describe its properties.
	The query $Q$ over $\Trg = \{\rel{G}/3\}$ is a Boolean query which is true (i.e., the empty tuple is its only output) if either $\rel{G}$ does not encode a function, i.e., it maps the same pair $(a,b)$ to different terms, or $\rel{G}$ does not encode an associative function,  or $\rel{G}$ does not encode a total function. In other words, $Q$ checks whether $\rel{G}$ does not encode a solution for $\mathsf{EMB}$.
	
	We are now ready to present the reduction.
	Let $A$ be a finite set and $f : A \times A \rightarrow A$ be a partial function. The reduction constructs the source instance $I_{A,f}$ and the empty tuple $\bt = ()$.
	Clearly, $I_{A,f}$ can be constructed in polynomial time w.r.t.\ $|I|$.
	It remains to show that $A,f$ is a ``yes''-instance of $\mathsf{EMB}$ iff $\bt \not \in \scert{I_{A,f}}{Q}{\Set}$.
	
	\emph{(Only if direction)} Assume $\bt \not \in \scert{I_{A,f}}{Q}{\Set}$. Then, there exists a supported solution $J \in \ssol{I_{A,f},\Set}$ of $I_{A,f}$ w.r.t.\ $\Set$ such that $\bt \not \in Q(J)$. By definition of supported solution, $J$ is finite and it only contains atoms with relation $\rel{G}$. Thus, by definition of $\Set$, $\bt \not \in Q(J)$ implies that $J$ necessarily encodes an extension of $f$, which is also total and associative.
	
	\emph{(If direction)} Assume $A,f$ is a ``yes''-instance of $\mathsf{EMB}$, and let $B \supseteq A$ be a finite set, and $g : B \times B \rightarrow B$ be the total associative function that extends $f$. Then, consider the instance $J$ over $\Trg$ defined as $J = \{\rel{G}(a,b,c) \mid a,b,c \in B \text{ and } g(a,b) = c \}$. It is not difficult to verify that $J$ is a supported solution of $I_{A,f}$ w.r.t.\ $\Set$. Finally, by construction of $J$, $\bt \not \in Q(J)$ as needed.
\end{proof}

Although the complexity result above tells us that computing supported certain answers might be infeasible in some settings, we can show that for weakly-acyclic settings, the complexity is more manageable. In particular, we prove that in this case, the problem is in $\co\NP$ and that this complexity bound is tight (i.e., there exist weakly-acyclic settings and queries for which the problem is $\co\NP\hard$).
We first focus on the upper bound.

\begin{theorem}\label{thm:scert-conp-upper}
	For every weakly-acyclic setting $\Set$ and every query $Q$, $\scertprob{\Set,Q}$ is in $\co\NP$.
\end{theorem}
\begin{proof}
	%
	
	We provide a non-deterministic polynomial-time procedure for solving the complement of the problem $\scertprob{\Set,Q}$, when $\Set$ is a weakly-acyclic data exchange setting. That is, given a source instance $I$ of $\Set$ and a tuple $\bt \in \Const^{\ar(Q)}$, the procedure non-deterministically constructs a supported solution $J^*$ of $I$ w.r.t.\ $\Set$ (if one exists), and checks whether $\bt \not \in Q(J^*)$. Let $\Set = \des$, and consider a source instance $I$ of $\Set$, a query $Q$ over $\Trg$, and a tuple $\bt \in \Const^{\ar(Q)}$.
	
	The procedure is defined in two parts. The first part is in charge of non-deterministically constructing a supported solution $J^*$. If the procedure was not able to construct a supported solution (i.e., no such solution exists, or it followed a wrong computation path), the procedure sets $J^* = \perp$. The second part simply verifies whether either $J^*=\perp$, in which case it rejects, or it checks whether $\bt \not \in Q(J^*)$, in which case accepts, otherwise rejects. The second part can be easily implemented by a deterministic polynomial-time procedure; we now show the first procedure constructing $J^*$.
	
	This procedure implements a variation of the so-called semi-oblivious chase algorithm; we refer the reader to~\cite{Marnette09} for more details. In the following, for each TGD $\rho \in \STTGD \cup \TTGD$, let $\rel{Chosen}_\rho$ be a fresh relation, not occurring in $\Src \cup \Trg$, of arity $|\fr(\rho)|$.
	
	\begin{enumerate}[align=parleft,left=0pt..1em]
		\item Let $J_0 = I$, and let the current step be $i = 0$.
		\item If $J_i$ does not satisfy the EGDs in $\TTGD$, then let $J^* = \perp$ and halt;
		\item If $J_i$ satisfies the EGDs in $\TTGD$, and no TGD $\rho \in \STTGD \cup \TTGD$ and homomorphism $h$ from $\body(\rho)$ to $J_i$ exist such that $\rel{Chosen}_\rho(h(\fr(\rho))) \not \in J_i$, then let $J^*$ be $J_i$ after removing all atoms over $\Src$ and the atoms using the $\rel{Chosen}$ predicates, and halt.
		
		\item Otherwise, guess a TGD $\rho_i \in \STTGD \cup \TTGD$ and a homomorphism $h_i$ from $\body(\rho_i)$ to $J_i$ such that $\rel{Chosen}_{\rho_i}(h_i(\fr(\rho_i))) \not \in J_i$, and guess an extension $h_i'$ of $h_i$ such that, for each $z \in \exvar(\rho_i)$, $h_i'(z)=c^i_z$, where either $c^i_z$ is a constant occurring in one of $\Set$, $I$, $Q$, or a fresh new constant. Finally, let $J_{i+1} = J_i \cup h_i'(\head(\rho_i)) \cup \{\rel{Chosen}_{\rho_i}(h_i(\fr(\rho_i))) \}$.
		Let $i := i + 1$ and goto 2.
	\end{enumerate}
	
	To show that the procedure above terminates after a polynomial number of steps, we can use a similar argument to the one given for proving Theorem~3.9 in~\cite{FaginKMP05}. We now show that, for every target instance $J$ of $\Set$, a run of the above procedure halting with $J^* = J$ exists iff $J$ is a supported solution of $I$ w.r.t.\ $\Set$, and the claim will follow. We focus on one of the two directions, as the other direction can be proved in a similar way.
	
	Assume there is a run of $n$ steps of the procedure above with $J^* = J$, for some target instance $J$ of $\Set$, and let $\rho_i$, $h_i$ and $c^i_z$, for $z \in \exvar(\rho_i)$ be the TGD, homomorphism and constants guessed at step $i$ in the run.
	Let $\gamma$ be the \exchoice\ such that, for each $i \in \{1,\ldots,n\}$, $\gamma(\rho_i,h_i(\fr(\rho_i))) = \{(z,c^i_z) \mid z \in \exvar(\rho_i)\}$.
	The fact that $\gamma$ is indeed an \exchoice\ follows from the fact that at each step $i \in \{1,\ldots,n\}$, a constant $c^i_z$ is introduced only if $\rel{Chosen}_{\rho_i}(h_i(\fr(\rho_i))) \not \in J_i$, which in turn implies that no constant has been chosen at some step $j<i$, where $h_j(\fr(\rho_j)) = h_i(\fr(\rho_i))$. By definition of the procedure, $J$ is the instance obtained from $J_n$ where all the atoms with relations in $\Src$ or of the form $\rel{Chosen}_{\rho}$ are removed.
	Hence, by construction, $I \cup J$ satisfies $\STTGD^\gamma$ and $J$ satisfies all the TGDs in $\TTGD^\gamma$. Since $J \neq \perp$, $J$ also satisfies the EGDs in $\TTGD^\gamma$. Moreover, no $J' \subsetneq J$ is such that $I \cup J'$ satisfies $\STTGD^\gamma$ and $J'$ satisfies $\TTGD^\gamma$. If this is the case, let $\alpha \in J \setminus J'$, and let $i \in \{1,\ldots,n\}$ be the step in the above run where $\alpha$ is added in $J_{i+1}$. Then, the TGD $\rho' = h_i(\rho_i) \rightarrow h_i'(\head(\rho_i))$ is in $\STTGD^\gamma \cup \TTGD^\gamma$, by construction of $\gamma$. However $J'$ does not satisfy $\rho'$. The latter, together with the previous discussion implies that $J$ is a supported solution of $I$ w.r.t.\ $\Set$. 
	%
\end{proof}

We point out that the above result is in contrast with all the data exchange semantics discussed in the introduction, for which computing certain answers is undecidable, even for weakly-acyclic settings~\cite{HernichLibkin2011,Hernich2011}.

\medskip
We now move to the lower bound and show that the $\co\NP$ upper bound is tight.

\begin{theorem}\label{thm:scert-conp-lower}
	There exists a weakly-acyclic setting $\Set$ that is TGD-only and a query $Q$ such that $\scertprob{\Set,Q}$ is $\co\NP\hard$.
\end{theorem}
\begin{proof}
	The $\co\NP\hard$ness is proved via a reduction from 3-colorability to the complement of our problem. 
	We encode the input graph $G = (V,E)$ as the instance
	\[
	\begin{array}{ll}
		I_G = & \{\Vrtx(u) \mid u \in V\} \cup \{\Edgs(u,v) \mid (u,v) \in E\} \cup \\ 
		& \{\Col(c) \mid c \in \{\mathsf{r},\mathsf{g},\mathsf{b}\}\}.
	\end{array}
	\]
	Colorings are constructed in the setting $\Set = \des$, via the source-to-target TGDs ($\TTGD$ is empty):
	\[
	\begin{array}{ll}
		\rho_1 & = \Col(x) \rightarrow \Colt(x),                    \\
		\rho_2 & = \Edgs(x,y) \rightarrow \Edgt(x,y),                           \\
		\rho_3 & = \Vrtx(x) \rightarrow \exists z\, \HasCol(x,z),
	\end{array}
	\]
	where $\Colt$ collects all colors, $\Edgt$ contains the edges of the graph in the target schema, and $\HasCol$ assigns a term to each node of the graph.
	
	The Boolean query $Q = Q_1 \vee Q_2$ is true over an instance of the target schema iff the instance does not encode a valid 3-coloring. In particular, $Q_1$ checks whether the ``color'' used for some node differs from $\mathsf{r},\mathsf{g}, \mathsf{b}$:
	$$Q_1 = \exists x,y\, \HasCol(x,y) \wedge \neg \Colt(y),$$
	whereas $Q_2$ checks whether the nodes of an edge have the same color:
	\begin{flalign*}
		&& Q_2 = \exists x,y,z\, \Edgt(x,y) \wedge \HasCol(x,z) \wedge \HasCol(y,z). &&
	\end{flalign*}
	
	We prove that $G$ admits a 3-coloring iff $\bt = () \not \in  \scert{I_G}{Q}{\Set}$.
	
	\medskip
	\emph{(Only if direction)} Assume $G$ admits a 3-coloring $\mu$ and consider the instance
	\[
	\begin{array}{ll}
		J=&\{\HasCol(v,\mu(v)) \mid v \in V\} \cup \{\Edgt(u,v) \mid (u,v) \in E\} \cup \{\Colt(c) \mid c \in \{\mathsf{r},\mathsf{g},\mathsf{b}\}\}.
	\end{array}
	\]
	It is not difficult to see that $J$ is a supported solution of $I_G$ w.r.t.\ $\Set$. Clearly, $\bt \not \in Q(J)$ and the claim follows.
	
	\emph{(If direction)} Assume that $G$ does not admit a 3-coloring, and consider an arbitrary supported solution $J$ of $I_G$ w.r.t.\ $\Set$.
	Note that for every edge $(u,v) \in E$, we have that $\Edgt(u,v) \in J$ and $\HasCol(u,c_1),\HasCol(v,c_2) \in J$, for some constants $c_1,c_2$.
	We distinguish two cases. Assume that there is an edge $(u,v) \in E$ such that $c_1 \not \in \{\mathsf{r},\mathsf{g},\mathsf{b}\}$ or $c_2 \not \in \{\mathsf{r},\mathsf{g},\mathsf{b}\}$. Thus, $\bt \in Q_1(J)$ which implies $\bt \in Q(J)$. Assume now that for every edge $(u,v) \in E$, $c_1,c_2 \in \{\mathsf{r},\mathsf{g},\mathsf{b}\}$. Thus, since $G$ does not admit a 3-coloring, for at least one edge $(u,v) \in E$, $c_1 =c_2$. Hence, $\bt \in Q_2(J)$, which implies that $\bt \in Q(J)$ and the claim follows.
\end{proof}

We point out that the query employed in the proof of the above theorem is a simple Boolean combination of CQs. This kind of FO queries have been studied in the context of incomplete databases, e.g., see \cite{GheerbrantL15}. However, differently from the incomplete databases setting, where such queries guarantee query answering in polynomial time, the complexity in our setting is higher, due to the presence of TGDs; the latter is true even for weakly-acyclic TGDs, as shown by Theorem~\ref{thm:scert-conp-lower} above. Similarly, arbitrary FO queries (e.g., involving also universal quantification) behave very differently depending on the given setting. For example, according to Theorem~\ref{thm:scert-conp-upper}, for \emph{any} FO query, supported certain answers remain in $\co\NP$, under weakly-acyclic settings, while for arbitrary settings, the use of universal quantification makes supported certain answering undecidable; the latter is a consequence of the proof of Theorem~\ref{thm:scert-undecidable}. Hence, one cannot directly conclude much on the complexity of supported certain answers by considering the query alone, as done for querying incomplete databases.

\medskip
We conclude this section by recalling that for positive queries, supported certain answers coincide with the classical ones (Theorem~\ref{thm:equiv-positive}), and computing (classical) certain answers for weakly-acyclic settings, under positive queries, is tractable~\cite{FaginKMP05}. Hence, the result below follows.

\begin{corollary}
	For every weakly-acyclic setting $\Set$ and every positive query $Q$, $\scertprob{\Set,Q}$ is in $\PTIME$.
\end{corollary}

\section{Exact Query Answering via Logic Programming}\label{sec:logic-programming}

In this section, we show how to compute supported certain answers exactly by means of a translation into logic programming under the stable model semantics, i.e., \emph{Answer Set Programming} (ASP). First, we need to recall the syntax and semantics of logic programs. In particular, we focus on a fragment of logic programs that is enough for our purposes, which is Datalog  with (possibly non-stratified) negation, which means we do not allow for function symbols or disjunctive rules.

\smallskip
\noindent
\textbf{Syntax.} A \emph{literal} $L$ is an expression of the form $\alpha$ or $\neg \alpha$, where $\alpha$ is either an atom without nulls, or the expression $t_1 = t_2$, where $t_1,t_2$ are variables or constants; we write $t_1 \neq t_2$ for $\neg t_1 = t_2$. We say that $L$ is \emph{positive} (resp., \emph{negative}) if $L = \alpha$ (resp., $L = \neg \alpha$).
If a literal contains no variables, it is said to be \emph{ground}.

A \emph{\lprule} $r$ is an expression of the form
$$ H \lprulesep A_1,\ldots,A_n, \neg B_1,\ldots,\neg B_m .$$
with  $n \ge 0$, $m \ge 0$, and where $H$ is either a positive literal or the symbol $\perp$, $A_1,\ldots,A_n$ are positive literals, and $\neg B_1,\ldots,\neg B_m$ are negative literals.
We denote $\head(r) = \{H\}$ as the \emph{head} of $r$, while $\body(r) = \{A_1,\ldots,A_n, \neg B_1,\ldots,\neg B_m\}$ is the \emph{body} of $r$; we use $\bodyp(r)$ to denote $\{A_1,\ldots,A_n\}$, and $\bodyn(r)$ to denote $\{B_1,\ldots,B_m\}$. 
If $H = \perp$, we say that $r$ is a \emph{constraint}. If $m=0$, we say the \lprule is \emph{positive}; if $r$ contains no variables, it is said to be \emph{ground}.
We say the \lprule $r$ is \emph{safe} if every variable in the rule occurs in some literal of $\bodyp(r)$. 
We will require every \lprule to be safe (besides being a common requirement, safe \lprules suffice for our purposes).
%

As customary, we will consider two kinds of sets of \lprules: 
\begin{enumerate}
\item finite sets of \lprules of the form $H \lprulesep$, with $H \neq \bot$ (notice that such \lprules must be ground because of safety), which are commonly used to  represent databases---a set of this kind will be called an \emph{\extdb};
\item finite sets of \lprules of any other form---a set of this kind will be called a \emph{\program}.
\end{enumerate}

\smallskip
\noindent
\textbf{Semantics.} Let $\P$ be a \program and $\EDB$ an \extdb.
We will often use $\PP$ to denote the set $\P \cup \EDB$.  
 The \emph{Herbrand universe} of $\PP$, denoted $\const(\PP)$, is the set of all constants occurring in $\PP$. The \emph{Herbrand base} of $\PP$, denoted $\base{\PP}$, is the set of all atoms that can be built using relations and constants occurring in $\PP$.
A \emph{ground version} of a \lprule $r \in \PP$ is a ground \lprule $r'$ that can be obtained from $r$ by replacing all occurrences of each variable $x$ of $r$ with some constant from $\const(\PP)$.

The \emph{grounding} of $\PP$, denoted $\ground{\PP}$, is the set of \lprules obtained from $\PP$ by replacing each \lprule $r \in \PP$ with all its ground versions.

We say that an instance $I$ \emph{satisfies} a ground positive literal $L$ if either $L$ is of the form $\alpha = \beta$ and $\alpha$ and $\beta$ are the same constant, or $L$ is an atom occurring in $I$. Furthermore, we say that $I$ satisfies a ground negative literal $\neg L$, if $I$ does not satisfy $L$. Finally, $I$ satisfies a set of ground literals if $I$ satisfies each literal in it.

Consider a rule $r \in \ground{\PP}$ and an instance $I$. We say that $I$ \emph{satisfies} $r$ if, either $r$ is a constraint and $I$ does not satisfy $\body(r)$, or $I$ satisfies $\body(r)$ implies that $I$ satisfies $\head(r)$ (notice that an empty body is always satisfied).

A \emph{model} of $\PP$ is an instance $M$ such that $M \subseteq \base{\PP}$ and such that $M$ satisfies each rule of $\ground{\PP}$. We say that $M$ is \emph{minimal} if there is no other model $M'$ of $\P$ such that $M' \subsetneq M$. We use $\MM{\PP}$ to denote the set of all minimal models of $\P$.

The \emph{reduct} of $\PP$ w.r.t.\ some instance $I$ is the set of ground \lprules obtained from $\ground{\PP}$ by removing each \lprule $r$ for which $I$ does not satisfy $\bodyn(r)$, and by removing all negative literals from the body of each \lprule $r$ for which $I$ satisfies $\bodyn(r)$.

An instance $M$ is a \emph{stable model} of $\PP$ if $M \in \MM{\PP'}$, where $\PP'$ is the reduct of $\PP$ w.r.t.\ $M$. We use $\SM{\PP}$ to denote the set of all stable models of $\PP$.

\smallskip
\noindent
\textbf{Cautious Reasoning.} Consider an \extdb $\EDB$, a \program $\P$, and a query $Q$. The \emph{cautious answers} to $Q$ over $\EDB$ and $\P$ is the set:
$$ \cans{\EDB}{Q}{\P} = \bigcap_{M \in \SM{\P_{\EDB}}} Q(M).$$


%
%

The key task we are interested in, regarding logic programs, is computing cautious answers. In particular, we are interested in its data complexity, i.e., when the program and the query are fixed; as usual, we focus on the decision version of the problem. In the following, $\P$ and $Q$ denote some program and some query, respectively:
\begin{center}
	
	\framebox[12cm]{\begin{tabular}{ll}
			{\small PROBLEM} :  & $\cansprob{\P,Q}$
			\\
			{\small INPUT} :    & An \extdb $\EDB$ and a tuple $\bt \in \Const^{\ar(Q)}$. \\
			{\small QUESTION} : & Is $\bt \in \cans{\EDB}{Q}{\P}$?
	\end{tabular}}
\end{center}

It is well known that for every program $\P$ and every query $Q$, $\cansprob{\P,Q}$ is in $\co\NP$---e.g., see \cite{GrecoSZ95}. 


\smallskip
\noindent
\textbf{The choice construct.}
We now extend logic programs with an additional construct, called \emph{choice}. We point out that extending logic programs with the choice is purely for syntactic convenience, as this construct can be implemented by means of standard {\lprule}s with negation.

The choice construct has been introduced in Datalog in \cite{SaccaZ90}, studied in \cite{GiannottiPSZ91,GrecoSZ95,GrecoZ98,GrecoZG92}, and implemented in the Datalog systems LDL++ \cite{ArniOTWZ03} and, in some form, in recent ASP systems (e.g., Potassco \cite{GebserKKOSS11} and DLV \cite{AlvianoFLPPT10}). It is used to enforce functional dependency (FD) constraints on rules of a logic program. 

A \emph{choice \lprule} $r$ is an expression of the form
$$ H \lprulesep A_1,\ldots,A_n, \neg B_1,\ldots,\neg B_m,\choice((X),(Y)).$$
where $n$, $m$, $H$, $A_1,\ldots,A_n$, and $B_1,\ldots,B_m$ are all defined as for standard {\lprule}s, while
 $X$ and $Y$ denote \emph{disjoint} sets of variables occurring in $\body(r)$.\footnote{When $X$ (resp., $Y$) is a singleton, we may use its only element in place of $X$ (resp., $Y$).}
The original definition of choice rule allows for multiple choice constructs in the rule body; here we focus on choice rules with only one choice construct in the body as this is enough for our purposes.

Intuitively, the construct $\choice((X),(Y))$ prescribes that the set of all consequences derived from $r$ must respect the functional dependency $X \rightarrow Y$. 


The formal semantics of choice rules is given in terms of a translation to standard {\lprule}s using negation. In particular, the choice \lprule $r$ defined above is a shorthand for writing the following set of {\lprule}s; in what follows, $\bx$ and $\by$ are the tuples  of all variables in $X$ and $Y$, respectively, in some arbitrary order.

\[
\begin{array}{ll}
	r^{(1)}: & \rel{Range}_r(\by) \lprulesep A_1,\ldots,A_n, \neg B_1,\ldots,\neg B_m. \\
	r^{(2)}: & H \lprulesep A_1,\ldots,A_n, \neg B_1,\ldots,\neg B_m, \rel{Chosen}_r(\bx, \by). \\
	r^{(3)} : & \rel{Chosen}_r(\bx, \by) \lprulesep A_1,\ldots,A_n, \neg B_1,\ldots,\neg B_m, \neg \rel{DiffChoice}_r(\bx, \by). \\
	r^{(4)}_i : & \rel{DiffChoice}_r(\bx, \by) \lprulesep \rel{Chosen}_r(\bx, \bw), \rel{Range}_r(\by), \by[i] \neq \bw[i],\ \forall i \in \{1,\ldots,|Y|\}. 
\end{array}
\]
In the above {\lprule}s, $\rel{Range}_r$, $\rel{Chosen}_r$, and $\rel{DiffChoice}_r$ are fresh relations not occurring in $\P$, which are used only to rewrite the \lprule $r$.

\subsection*{Implementing Supported Certain Answers via Logic Programming with Choice}

The goal of this section is to prove the following key result.

%

\begin{theorem}\label{thm:lp-reduction}
	For every weakly-acyclic data exchange setting $\Set = \des$, and every query $Q$ over $\Trg$, there exists a program $\P$ such that $\scertprob{\Set,Q}$ reduces to $\cansprob{\P,Q}$ in polynomial time.
\end{theorem}

The rest of this section is devoted to prove the above claim. In particular, we show how to convert a weakly-acyclic data exchange setting $\Set = \des$, together with a source instance $I$ of $\Set$ and a query $Q$ over $\Trg$, into an \extdb $\EDB$ and a \program $\P$ using choice rules, in such a way that $\P$ depends only on $\Set$ and such that $\scert{I}{Q}{\Set} = \cans{\EDB}{Q}{\P}$.


\medskip
The main idea of the translation is to derive a program together with an \extdb such that the stable models correspond to a subset of the supported solutions that is enough for computing supported certain answers. For this, we rely on the following useful result that one can extract from the proof of Theorem~\ref{thm:scert-conp-upper}.
For a set $S$ of terms and a set of instances $\mathcal{I}$, we use $\restr{\mathcal{I}}{S}$ to denote the set of instances $\{I \in \mathcal{I} \mid \adom(I) \subseteq S\}$.

\begin{lemma}\label{lem:finite-query}
	Consider a weakly-acyclic data exchange setting $\Set = \des$. There exists a polynomial $\mathsf{pol}$ such that, for every source instance $I$ of $\Set$, and every query $Q$ over $\Trg$, the following holds:
	$$ \scert{I}{Q}{\Set} = \bigcap\limits_{J \in \restr{\ssol{I,\Set}}{S}} Q(J),$$
	where $S$ is the set of all constants occurring in $\Set$, $I$ and $Q$, plus some fixed, arbitrarily chosen constants $c_1,\ldots,c_{\mathsf{pol}(|I|)}$ not occurring anywhere in $\Set$, $I$, or $Q$.
\end{lemma}
\begin{proof}
	The claim easily follows by construction of the non-deterministic procedure building the instance $J^*$ in the proof of Theorem~\ref{thm:scert-conp-upper}, from the fact that it terminates after a polynomial number of steps w.r.t.\ $I$, and the fact that it halts with $J^* \neq \perp$ iff $J^*$ is a supported solution in $\ssol{I,\Set}$.
\end{proof}

The result above tells us that considering supported solutions of a certain polynomial size suffices for computing supported certain answers. 
The stable models of the program together with the \extdb we are going to define will correspond to such supported solutions.

%
%

\begin{definition}[Translation]\label{def:translation}
	Consider a data exchange setting $\Set = \des$,  a source instance $I$ of $\Set$, a query $Q$ over $\Trg$, and the set of constants $S$ as defined in Lemma~\ref{lem:finite-query} w.r.t.\ $\Set$, $I$ and $Q$.
	
	\smallskip
	We use $\LP{\Set}$ to denote the set consisting of the following {\lprule}s.
	\begin{enumerate}[align=parleft,left=0pt..1em]
		
		\item For each TGD $\rho$ of the form $\alpha_1 \wedge \cdots \wedge \alpha_n \rightarrow \exists \bz\, \beta_1 \wedge \cdots \wedge \beta_m$ in $\STTGD \cup \TTGD$, with $\by = \fr(\rho)$, if $k = |\bz| = 0$, the following \lprules are introduced:
\begin{equation}\label{eq:tgd1}
\beta_i \lprulesep \alpha_1,\ldots,\alpha_n,\ \ \ \ i \in \{1,\ldots,m\}, 
\end{equation}
		otherwise, the following \lprules are introduced:
\begin{equation}\label{eq:tgd2a}
\rel{ExChoice}_\rho(\by,\bz) \lprulesep \alpha_1,\ldots,\alpha_n, \rel{Dom}(\bz[1]),\ldots,\rel{Dom}(\bz[k]),\choice((Y),(Z)),   
\end{equation}
\begin{equation}\label{eq:tgd2b}
\beta_i \lprulesep \rel{ExChoice}_\rho(\by,\bz),  \ \ \ \ i \in \{1,\ldots,m\}, 
\end{equation}
		
		where $Y$ and $Z$ are the sets of all variables in $\by$ and $\bz$, respectively, and $\rel{Dom}$ is a fresh predicate.
		
		\item For each EGD $\alpha_1 \wedge \cdots \wedge \alpha_n \rightarrow x = y$ in $\TTGD$, the following constraint is introduced:
\begin{equation}\label{eq:egd}
\perp \lprulesep \alpha_1,\ldots,\alpha_n, x \neq y	
\end{equation}

\end{enumerate}
	
\smallskip
We use $\ED{\Set,I,Q}$ to denote the extensional database consisting of the following {\lprule}s.
	\begin{enumerate}[align=parleft,left=0pt..1em]
		\item For each constant $c \in S$, the following \lprule is introduced:
		\begin{equation}\label{eq:S}
			\rel{Dom}(c) \lprulesep .
		\end{equation}
		
		\item For each fact $\alpha \in I$, the following \lprule is introduced:
			\begin{equation}\label{eq:fact}
				\alpha \lprulesep .
			\end{equation}
	\end{enumerate}
\end{definition}

\begin{example}
Considering the data exchange setting $\Set$ and the source instance $I$ of $\Set$ from Example~\ref{ex:employee}, we have that $\LP{\Set}$ is the following logic program:
\[
\begin{array}{l}
\rel{ExChoice}_{\rho_1}(x,z) \lprulesep \Emp(x),\ \rel{Dom}(z),\choice((x),(z)).  \\
\EmpNew(x,z) \lprulesep \rel{ExChoice}_{\rho_1}(x,z).  \\\ \\
\EmpNew(x,y) \lprulesep \EmpC(x,y). \\
\SameC(x,x') \lprulesep \EmpNew(x,y),\ \EmpNew(x',y). \\
\perp \lprulesep \EmpNew(x,y),\ \EmpNew(x,z),\  y \neq z.\\
\end{array}
\]	
Intuitively, the choice rule associated to the TGD $\rho_1$ is in charge of non-deterministically assigning a certain value to the existential variables of $\rho_1$, for each value its frontier variables can take, i.e., the choice rule essentially builds an \exchoice\ for $\rho_1$. Once the \exchoice\ is constructed, the rule $\EmpNew(x,z) \lprulesep \rel{ExChoice}_{\rho_1}(x,z)$ simply propagates these choices to the head of $\rho_1$, as needed. All other TGDs have no existential quantification, and so use no choice construct. Finally, the only EGD $\eta$ is converted to a constraint, so that the stable models of the logic program satisfy $\eta$.
\end{example}

We are now ready to prove Theorem~\ref{thm:lp-reduction}.\\

\textbf{Proof of Theorem~\ref{thm:lp-reduction}.}
Given an instance $I$ over a schema $\Sch$, and a schema $\Sch' \subseteq \Sch$, we use $I[\Sch']$ to denote the restriction of $I$ to only its facts referring to relations in $\Sch'$.
Notice that for every query $Q$ over $\Sch'$, the following holds: $Q(I)=Q(I[\Sch'])$.

Consider a data exchange setting $\Set = \des$,  a source instance $I$ of $\Set$, a query $Q$ over $\Trg$, and the set of constants $S$ as defined in Lemma~\ref{lem:finite-query} w.r.t.\ $\Set$, $I$, and $Q$.

Let $\P=\LP{\Set}$ and $\EDB=\ED{\Set,I,Q}$.

We want to show 
$ \cans{\EDB}{Q}{\P} = \scert{I}{Q}{\Set}.$
Leveraging Lemma~\ref{lem:finite-query},
we show that
$\{M[\Trg] \mid M \in \SM{\PP}\} = \restr{\ssol{I,\Set}}{S}.$


\medskip
(1) In the following, we show $\{M[\Trg] \mid M \in \SM{\PP}\} \subseteq \restr{\ssol{I,\Set}}{S}.$
Let $X \in \{M[\Trg] \mid M \in \SM{\PP}\}$ and $M$ be a stable model in $\SM{\PP}$ such that $X=M[\Trg]$.

Let $\gamma$ be an \exchoice\  defined as follows: given a TGD $\rho = \varphi(\x,\y) \rightarrow \exists \z\, \psi(\y,\z)$ 
in $\STTGD\cup \TTGD$ and a tuple $\bt \in \Const^{|\y|}$, $\gamma$ returns a set $\gamma(\rho,\bt)$ of pairs  of the form $(z_i,c)$, one for each existential variable $z_i \in \z$, where $c$ 
is defined as follows:
if $\rel{ExChoice}_\rho(\bt,c_1,\dots,c_k)\in M$, then $c=c_i$, otherwise $c$ is an arbitrary constant of $S$.

It is easy to see that $\const(\PP)=S$, and thus $X$ contains only constants in $S$.
%
Moreover, $I \cup X$  satisfies $\STTGD^\gamma$, because otherwise $M$ would not satisfy some ground version of the \lprules derived from the TGDs in $\STTGD^\gamma$.
Also, $X$ satisfies $\TTGD^\gamma$, because otherwise $M$ would not satisfy some ground version of the \lprules derived from the TGDs/EGDs in $\TTGD^\gamma$.

Since every stable model is also a minimal model, the minimality of $M$ ensures that there is no  $J' \subsetneq X$  such that $I \cup J'$ satisfies $\STTGD^\gamma$ and $J'$ satisfies $\TTGD^\gamma$.
Thus, $X$ is a supported solution of $I$ w.r.t.\ $\Set$ containing only constants in $S$.


\medskip
(2) We now show $\{M[\Trg] \mid M \in \SM{\PP}\} \supseteq \restr{\ssol{I,\Set}}{S}.$
Let $J \in \restr{\ssol{I,\Set}}{S}$ and $\gamma$ be the \exchoice\ for which $I \cup J$ satisfies $\STTGD^\gamma$ and $J$ satisfies $\TTGD^\gamma$.
Let $X = I \cup J \cup \{\rel{Dom}(c) \mid c \in S\}$.
We show that $X \in \SM{\PP}$. 

First, $X$ satisfies each ground rule in $\ground{\P}$ of the form $\beta_i \lprulesep \alpha_1,\ldots,\alpha_n$ (cf.~(\ref{eq:tgd1}) in Definition~\ref{def:translation}), because otherwise the TGD of the form $\alpha_1 \wedge \cdots \wedge \alpha_n \rightarrow  \beta_1 \wedge \cdots \wedge \beta_m$ in $\STTGD^\gamma$ or $\TTGD^\gamma$ would not be satisfied by $I \cup J$ or $J$, respectively.

Also, $X$ satisfies the ground rules in $\ground{\P}$ of the form~(\ref{eq:tgd2a})--(\ref{eq:tgd2b}) in Definition~\ref{def:translation}, derived from a TGD having existential variables, because otherwise such a TGD would not be satisfied by either $I \cup J$ or $J$, or $J$ would not be minimal.

Further, $X$ satisfies each ground constraint of the form $\perp \lprulesep \alpha_1,\ldots,\alpha_n, x \neq y$ (cf. (\ref{eq:egd}) in Definition~\ref{def:translation}) in $\ground{\P}$ as otherwise $J$ would not satisfy the EGD $\alpha_1 \wedge \cdots \wedge \alpha_n \rightarrow x = y$ in $\TTGD^\gamma$.

Then, $X$ satisfies each rule in $\EDB$ of the form~(\ref{eq:S}) of Definition~\ref{def:translation} because $\{\rel{Dom}(c) \mid c \in S\} \subseteq X$.

Finally, $X$ satisfies each rule in $\EDB$ of the form~(\ref{eq:fact}) of Definition~\ref{def:translation} because $X$ contains $I$.

By the minimality of $J$ we obtain the minimality of $X$, and thus, $X$ is a stable model of $\PP$.
Noting that $X[\Trg] =J$, we conclude that $J\in \{M[\Trg] \mid M \in \SM{\PP}\}$.\qquad$\Box$

\section{Approximate Query Answering via Materialization}\label{sec:univ-solutions}
As already discussed in the introduction, there might exist scenarios where it is desirable to materialize a target instance starting from the source instance and the schema mapping, in such a way that supported certain query answers can be computed by considering the target instance alone. The goal of this section is thus to study the problem of materializing such an instance, when focusing on our notion of supported solutions.

%
%

It would be very useful if such a special target instance could be computed in polynomial-time, already for weakly-acyclic settings. However, due to Theorem~\ref{thm:scert-conp-lower}, this would imply $\PTIME=\co\NP$. Hence, we need something different.


We introduce a special instance that enjoys the following properties: the answers over this instance are an approximation (i.e., a subset) of the supported certain answers for general queries, but they coincide with supported certain answers for positive queries. We also show that we can compute such an instance in polynomial time for weakly-acyclic settings.

Our approach relies on conditional instances~\cite{ImielinskiL84}, which we introduce in the following.

\smallskip
\noindent
\textbf{Conditional instances.}
%
%
A \emph{valuation} $\nu$ is a mapping from $\Const \cup \Null$ to $\Const$ that is the identity on $\Const$.
A \emph{condition} $\phi$ is an expression that can be built using the standard logical
connectives $\wedge$, $\vee$, $\neg$, $\Rightarrow$, and expressions of the form
$t = u$, where $t,u \in \Const \cup \Null$. We will also use $t \neq u$ as a shorthand for $\neg (t = u)$. 
We write $\nu \models \phi$ to state that $\nu$ satisfies $\phi$,
%
and $\phi \models \psi$ if all valuations satisfying $\phi$ satisfy the condition $\psi$.
%
A \emph{conditional fact} is a pair $\catom{\alpha,\phi}$, where $\alpha$ is a fact and $\phi$ is a condition.
%
%
A \emph{conditional instance} $\cinst{I}$ is a finite set of conditional facts. We also denote $\cinst{I}[1] = \{\alpha \mid \catom{\alpha,\phi} \in \cinst{I}\}$. A \emph{possible world} of a conditional instance $\cinst{I}$ is an instance $I$ such that there exists a valuation $\nu$ with $I = \{ \nu(\alpha) \mid \< \alpha, \phi\> \in \cinst{I} \text{ and } \nu \models \phi\}$. We use $\pw{\cinst{I}}$ to denote the set of all possible worlds of $\cinst{I}$.

\begin{definition}
Consider a conditional instance $\cinst{I}$ and a query $Q$. The \emph{conditional certain answers} of $Q$ over $\cinst{I}$ is the set $\ccert{\cinst{I}}{Q} = \bigcap_{J \in \pw{\cinst{I}}} Q(J)$.
\end{definition}
We are now ready to introduce our main tool.

\begin{definition}[Approximate Conditional Solution]\label{def:univ-cond-sol}
	Consider a data exchange setting $\Set$ and a source instance $I$ of $\Set$. A conditional instance $\cinst{J}$ is an \emph{approximate conditional solution} of $I$ w.r.t.\ $\Set$, if for every query $Q$:
	\begin{enumerate}[align=parleft,left=0pt..1em]
		\item $\ssol{I,\Set} \subseteq \pw{\cinst{J}}$, and thus $\ccert{\cinst{J}}{Q} \subseteq \scert{I}{Q}{\Set}$, and 
		\item if $Q$ is positive, $\ccert{\cinst{J}}{Q} = \scert{I}{Q}{\Set}$.
	\end{enumerate}
\end{definition}

That is, an approximate conditional solution is a conditional instance that allows to compute approximate answers for general queries, and exact answers for positive queries.


It is easy to observe that there are settings $\Set = \des$ and source instances $I$ for which an approximate conditional solution might not exist, even if $\Set$ is weakly-acyclic. This is due to the presence of EGDs in $\TTGD$. 

However, for weakly-acyclic settings without EGDs, an approximate conditional solution always exists, and we present a polynomial-time algorithm that is able to construct one.
We show how to deal with general weakly-acyclic settings with EGDs in Section~\ref{sec:egds}.

The algorithm is a variation of the well-known chase algorithm, which iteratively introduces new facts, starting from a source instance, whenever a TGD is not satisfied, i.e., it triggers the TGD. This variation also allows for a conditional triggering of TGDs, where new atoms are introduced, under the condition that some terms in the body coincide.

\smallskip
\textbf{Normal TGDs.}
To simplify the discussion, we consider an extension of TGDs that allow for equality predicates in the body. We will use these TGDs to rewrite standard TGDs in the following normal form.
A \emph{normal form TGD} $\rho$ is an expression of the form $\varphi(\x,\y) \wedge \eta(\x,\y) \rightarrow \exists \z\, \psi(\y,\z)$, where $\varphi$ and $\psi$ are conjunctions of atoms, $\varphi$ uses only variables and each variable in $\varphi$ occurs once in $\varphi$. The formula $\eta$ is a conjunction of equalities of the form $x=t$, where $x$ is a variable in $\x$ or $\y$, and $t$ is either a variable in $\x$ or $\y$, or a constant. The above equalities denote which variables should be considered to be the same and which positions should contain a constant. A (set of) standard TGDs $\Sigma$ can be converted in normal form in the obvious way. We denote $\norm{\Sigma}$ as the (set of) TGDs in normal form obtained from $\Sigma$.

\smallskip

In the following, fix a conditional instance $\cinst{I}$, a TGD $\rho$ with $\norm{\rho} = \varphi(\x,\y) \wedge \eta(\x,\y) \rightarrow \exists \z\, \psi(\y,\z)$, and a homomorphism $h$ from $\varphi(\x,\y)$ to $\cinst{I}[1]$. We use $h(\eta(\x,\y))$ to denote the condition obtained from $\eta(\x,\y)$ by replacing each variable $x$ therein with $h(x)$. Letting $h(\varphi(\x,\y)) = \{\alpha_1,\ldots,\alpha_n\}$, we use $\Cond{\rho,h}{\cinst{I}}$ to denote the set of all conditions of the form $h(\eta(\x,\y)) \wedge \phi_1 \wedge \cdots \wedge \phi_n$, such that $\catom{\alpha_i,\phi_i} \in \cinst{I}$, for each $i \in \{1,\ldots,n\}$.

\begin{example}\label{ex:conditions}
	Consider the TGD $\rho_3$ of Example~\ref{ex:employee}. The normal form TGD $\norm{\rho_3}$ is
	$$\EmpNew(x,y), \EmpNew(x',y'), y=y' \rightarrow \SameC(x,x').$$
	Consider now the conditional instance
	$$\cinst{I}=\{\catom{\EmpNew(\rel{\john},\rel{miami}),\nul_1 = a}, \catom{\EmpNew(\rel{mary},\nul_2),\mathsf{true}}\},$$
	where $a$ is a constant.
	Then, the mapping $h = \{x/\rel{john}, y/\rel{miami},x'/\rel{mary}, y'/\nul_2\}$ is a homomorphism from $\{\EmpNew(x,y), \EmpNew(x',y')\}$ to $\cinst{I}[1]$. Moreover, $\Cond{\rho_3,h}{\cinst{I}} = \{\nul_2 = \rel{miami} \wedge \nul_1 = a  \}$.
\end{example}

We are now ready to define the notion of conditional chase step.
In what follows, for a conditional instance $\cinst{I}$, a TGD $\rho$  with $\norm{\rho} = \varphi(\x,\y) \wedge \eta(\x,\y) \rightarrow \exists \z\, \psi(\y,\z)$ and a homomorphism $h$ from $\varphi(\x,\y)$ to $\cinst{I}[1]$, we use $\mathsf{result}(\cinst{I},\rho,h)$ to denote the set of atoms obtained from $\head(\norm{\rho})$, where each frontier variable $x$ in $\fr(\norm{\rho})$ is replaced with $h(x)$, and each existential variable $z$ in $\exvar(\norm{\rho})$ is replaced with a fresh null not occurring in $\cinst{I}$.

\begin{definition}[Conditional Chase Step]\label{def:cond-chase-step}
    Consider a conditional instance $\cinst{I}$, a TGD $\rho$, and let $\norm{\rho} = \varphi(\x,\y) \wedge \eta(\x,\y) \rightarrow \exists \z\, \psi(\y,\z)$. A \emph{conditional chase step of $\cinst{I}$ w.r.t.\ $\rho$} is an expression of the form $\cinst{I} \Gchase{\rho}{h}{\phi} \cinst{J}$,
    where \emph{(i)} $h$ is a 
	homomorphism from $\varphi(\x,\y)$ to $\cinst{I}[1]$, \emph{(ii)} $\phi \in \Cond{\rho,h}{\cinst{I}}$ is such that $\phi \not \models \false$, and \emph{(iii)} $\cinst{J} = \cinst{I} \cup \{\catom{\alpha,\phi} \mid \alpha \in \mathsf{result}(\cinst{I},\rho,h)\}$.
\end{definition}

\begin{example}
	Consider the conditional instance $\cinst{I}$, the homomorphism $h$ and the TGD $\rho_3$ of Example~\ref{ex:conditions}. Then, $\cinst{I} \Gchase{\rho_3}{h}{\phi} \cinst{J}$ is a conditional chase step, where $\phi$ is the condition
	$\nul_2 = \rel{miami} \wedge \nul_1 = a$,
	and $\cinst{J} = \cinst{I} \cup \{\catom{\SameC(\rel{john},\rel{mary}),\phi}\}$.
\end{example}

With the notion of conditional chase step at hand, we can define conditional chase sequences, which are sequences of conditional chase steps. For this we need one additional notion. A \emph{conditional tuple} is a pair $\catom{\bt,\phi}$, where $\bt$ is a tuple of constants and nulls, and $\phi$ a condition. For two conditional tuples $\catom{\bt,\phi},\catom{\bu,\psi}$, with $|\bt| = |\bu| = n$, we write $\catom{\bt,\phi}  \sqsubseteq \catom{\bu,\psi}$ if $\phi \models \psi$ and $\phi \models \bt = \bu$, where $\bt = \bu$ is a shorthand for the condition $\bigwedge^n_{i=1} \bt[i] = \bu[i]$. We write $\catom{\bt,\phi}  \not \sqsubseteq \catom{\bu,\psi}$, if $\catom{\bt,\phi}  \sqsubseteq \catom{\bu,\psi}$ does not hold.

Intuitively, $\catom{\bt,\phi}, \catom{\bu,\psi}$ should be understood to be two tuples, each of them belonging to a set of ``worlds'', described by the valuations that satisfy their conditions. Moreover, $\catom{\bt,\phi}  \sqsubseteq \catom{\bu,\psi}$ means that every world in which $\bt$ occurs, is also a world in which $\bu$ occurs ($\phi \models \psi$), and in each such world, $\bt$ and $\bu$ are the same tuples.

\begin{definition}[Conditional Chase Sequence]\label{def:cond-chase-seq}
    Consider a TGD-only data exchange setting $\Set = \des$ and a source instance $I$ of $\Set$. A \emph{conditional chase sequence of $I$ w.r.t.\ $\Set$} is a (possibly infinite) sequence of conditional instances $(\cinst{J}_i)_{i \ge 0}$, where for each $i \ge 0$,
    $\cinst{J}_i \Gchase{\rho_i}{h_i}{\phi_i} \cinst{J}_{i+1}$,
    and \emph{(i)} $\cinst{J}_0 = \{\<\alpha,\true\> \mid \alpha \in I\}$, \emph{(ii)} $\rho_i \in \STTGD \cup \TTGD$, for $i \ge 0$, and \emph{(iii)} for every $j < i$, if $\rho = \rho_i = \rho_j$, then $\catom{h_i(\fr(\rho)),\phi_i} \not \sqsubseteq \catom{h_j(\fr(\rho)),\phi_j}$.
\end{definition}

Intuitively, condition $\emph{(iii)}$ of the definition above is required to prevent the chase sequence to apply superfluous steps. That is, at a certain step, a fact is produced only if the possible worlds in which it occurs is not a subset of the possible worlds in which the same fact has already been introduced by previous steps. An example follows.

\begin{example}\label{ex:chase-sequence}
	Consider the data exchange setting $\Set = \des$, with $\Src = \{A/1,B/1\}$, $\Trg = \{R/2,S/1,T/1\}$, where the sets $\STTGD = \{\rho_1,\rho_2\}$ and $\TTGD=\{\rho_3\}$ are such that $\rho_1 = A(x) \rightarrow \exists z\, R(x,z)$, $\rho_2 = B(x) \rightarrow  S(x)$, and $\rho_3 =  R(x,y), S(y) \rightarrow T(x)$.
	Given $I = \{ A(a), B(b_1),B(b_2) \}$, the following is a conditional chase sequence of $I$ w.r.t.\ $\Set$:
	\begin{flalign*}
		&&\cinst{J}_0 & = \{\catom{A(a),\mathsf{true}},\catom{B(b_1),\mathsf{true}},\catom{B(b_2),\mathsf{true}}\}, & \cinst{J}_1 &= \cinst{J}_0 \cup \{ \catom{R(a,\nul),\mathsf{true}}\}, &&\\
		&& \cinst{J}_2 &= \cinst{J}_1 \cup \{ \catom{S(b_1),\mathsf{true}}\}, &\cinst{J}_3 &= \cinst{J}_2 \cup \{ \catom{S(b_2),\mathsf{true}}\},&& \\
		&& \cinst{J}_4 &= \cinst{J}_3 \cup \{ \catom{T(a),\nul = b_1}\}, &\cinst{J}_5 &= \cinst{J}_4 \cup \{ \catom{T(a),\nul = b_2}\}.&&
	\end{flalign*}
\end{example}

For a TGD-only setting $\Set = \des$  and a source instance $I$ of $\Set$, a \emph{finite} conditional chase sequence $(\cinst{J}_i)_{0 \le i \le n}$ of $I$ w.r.t.\ $\Set$ is \emph{maximal} if there is no conditional instance $\cinst{J}_{n+1}$, such that $(\cinst{J}_i)_{0 \le i \le n+1}$ is a conditional chase sequence of $I$ w.r.t.\ $\Set$. We call $\cinst{J}_n$ the \emph{result} of the maximal sequence.

\begin{example}\label{ex:maximal-sequence}
	Consider the conditional chase sequence $\cinst{J}_0,\ldots,\cinst{J}_5$ of Example~\ref{ex:chase-sequence}. The sequence is maximal, since any conditional chase step of the form $\cinst{J}_5 \Gchase{\rho}{h}{\phi} \cinst{J}$, for some conditional instance $\cinst{J}$, cannot satisfy condition \emph{(iii)} of Definition~\ref{def:cond-chase-seq}. The sequence $\cinst{J}_0,\ldots,\cinst{J}_4$ is not maximal because although a conditional atom of the form $\catom{T(a),\phi}$ is already present in $\cinst{J}_4$, an additional conditional atom of the same form needs to be introduced in $\cinst{J}_5$.
	This is needed to allow the fact $T(a)$ to be present for two different reasons (either because $\nul = b_1$ or $\nul = b_2$), and both reasons should occur in the result of the sequence.
\end{example}

We are now ready to present the main result of this section.
In what follows, given a schema $\Sch$ and a conditional instance $\cinst{I}$,  $\cinst{I}_{|\Sch}$ denotes the restriction of $\cinst{I}$ to its conditional facts with relations in $\Sch$.

\begin{theorem}\label{thm:chase-result}
    Consider a TGD-only setting $\Set = \des$ and a source instance $I$ of $\Set$. If $\cinst{J}$ is the result of a maximal conditional chase sequence of $I$ w.r.t.\ $\Set$, then $\cinst{J}_{|\Trg}$ is an approximate conditional solution of $I$ w.r.t.\ $\Set$.
\end{theorem}
\begin{proof}
	To prove the claim, it is enough to prove that each supported solution $J \in \ssol{I,\Set}$ is such that $I \cup J$ is a possible world of $\cinst{J}$, and that each possible world $J$ of $\cinst{J}$ contains a supported solution.
	We prove first that each $J \in \ssol{I,\Set}$ is such that $I \cup J$ is a possible world of $\cinst{J}$.
	
	Let $\gamma$ be the \exchoice\ witnessing that $J$ is a supported solution. 
	Then, $J$ can be characterized as the result of a procedure that computes a sequence $J_0, J_1, \dots J_m$ such that $J_0=I$, $J_m=J$, and each $J_i$ with $i>0$ is obtained from $J_{i-1}$ by adding the head of a TGD in $\STTGD^\gamma \cup \TTGD^\gamma$ whose body is contained in $J_{i-1}$ (i.e., the first part of the procedure in the proof of Theorem~\ref{thm:scert-conp-upper})---notice that such a procedure ensures also the minimality of $J$.
	For each step of the aforementioned procedure, there must be a corresponding conditional chase step in the sequence yielding $\cinst{J}$, which in turn induces $I \cup J$ as a possible world.
		
	Regarding whether each possible world $J$ of $\cinst{J}$ contains a supported solution, consider a possible world $J \in \pw{\cinst{J}}$. By construction of $\cinst{J}$, $J = I \cup J'$, for some instance $J'$ over $\Trg$, since all conditional facts in $\cinst{J}$, which correspond to the facts in $I$, have the always true condition.
	Moreover, by construction of $\cinst{J}$, $I \cup J'$ satisfies $\STTGD$, and $J'$ satisfies $\TTGD$. Hence, if we consider the set of TGDs $\STTGD^* \cup \TTGD^*$, where $\STTGD^*$ and $\TTGD^*$ are the sets of all ground versions of the TGDs in $\STTGD$ and $\TTGD$, respectively, we have that $I \cup J'$ satisfies $\STTGD^* \cup \TTGD^*$ and $J'$ satisfies $\TTGD^*$. However, since $\STTGD^\gamma \subseteq \STTGD^*$, and $\TTGD^\gamma \subseteq \TTGD^*$, for any \exchoice\ $\gamma$, we must have that $J'$ must contain a supported solution in $\ssol{I,\Set}$, as needed. 
\end{proof}

\begin{example}
	Consider the setting $\Set$, the source instance $I$ of $\Set$, and the conditional chase sequence $\cinst{J}_0,\ldots,\cinst{J}_5$ of Example~\ref{ex:chase-sequence}. From Theorem~\ref{thm:chase-result}, we conclude that $\cinst{J}_5$ is an approximate conditional solution for $I$ w.r.t.\ $\Set$.
\end{example}

We can further show that for TGD-only weakly-acyclic settings, a maximal conditional chase sequence always exists, and its length is polynomial. Moreover, its result can be computed in polynomial time.

\begin{theorem}\label{thm:ptime-universal}
	Consider a data exchange setting $\Set$ that is TGD-only and weakly-acyclic, and a source instance $I$ of $\Set$. Every conditional chase sequence $s = (\cinst{J}_i)_{0 \le i \le n}$ of $I$ w.r.t.\ $\Set$ is such that $n$ is a polynomial of $|I|$, and the result $\cinst{J}_n$ of $s$  can be computed in polynomial time w.r.t.\ $|I|$.
\end{theorem}
\begin{proof}
	To prove that the length of a conditional chase sequence is bounded by a polynomial, it suffices to follow an argument similar to the one given in~\cite{FaginKMP05} for proving that the length of a standard chase sequence is polynomial, for weakly-acyclic settings.
		Let $s = (\cinst{J}_i)_{0 \le i \le n}$ be a conditional chase sequence of $I$ w.r.t.\ $\Set$, with $\cinst{J}_i \Gchase{\rho_i}{h_i}{\phi_i} \cinst{J}_{i+1}$, for $i \in \{0,\ldots,n-1\}$. Since $n$ is a polynomial of $|I|$, we just need to show that for each $i \in \{0,\ldots,n-1\}$, $\cinst{J}_{i+1}$ can be constructed in polynomial time.
		To this end, it suffices to focus on condition (ii) of Definition~\ref{def:cond-chase-step} and condition (iii) of Definition~\ref{def:cond-chase-seq}. Since $n$ is polynomial, the maximum number of terms occurring in each condition $\phi_i$ is polynomial. Thus, each $\phi_i$ contains at most polynomially many equalities, and we can easily check whether $\phi_i \not \models \false$, by simply computing the closure of all equalities in $\phi_i$, and checking whether an equality of the form $a = b$ can be derived, where $a,b$ are distinct constants. Similarly, for each $i \in \{0,\ldots,n-1\}$, and every $j<i$, we can check whether $\catom{h_i(\fr(\rho)),\phi_i} \not \sqsubseteq \catom{h_j(\fr(\rho)),\phi_j}$, by using a similar approach.
\end{proof}



\medskip
\noindent
\textbf{Querying Approximate Conditional Solutions.}
%
What now remains to show is how we can compute the conditional certain answers over an approximate conditional solution, e.g., obtained via the conditional chase. It is known that the problem of computing the conditional certain answers of a query $Q$ is $\co\NP\hard$ in general, even when we assume all the conditions in the given conditional instance are true~\cite{ImielinskiL84}.
Hence, given a data exchange setting $\Set$ and a source instance $I$ of $\Set$, if an approximate conditional instance $\cinst{J}$ of $I$ w.r.t.\ $\Set$ can be computed in polynomial time w.r.t.\ $|I|$, one cannot always compute $\ccert{\cinst{J}}{Q}$, in polynomial time. Hence, we require an additional step of approximation.

%

To this end, we exploit an existing algorithm presented in~\cite{GrecoMT19} to compute approximate certain answers over incomplete databases.  Here we only recall the main properties of the algorithm. For more details, we refer the reader to~\cite{GrecoMT19}.
%

%
For a query $Q$, the function $\conditional{Q}_t$ from conditional instances to sets of tuples is defined in~\cite{GrecoMT19}, and it is such that the following holds.

\begin{theorem}
	Given a conditional instance $\cinst{J}$ over some schema $\Sch$ and a query $Q$ over $\Sch$, then
	\begin{enumerate}[align=parleft,left=0pt..1em]
		\item $\conditional{Q}_t(\cinst{J}) \subseteq \ccert{\cinst{J}}{Q}$;
		\item if $Q$ is positive, $\conditional{Q}_t(\cinst{J}) = \ccert{\cinst{J}}{Q}$;
		\item if every condition in $\cinst{J}$ is a conjunction of equalities, then $\conditional{Q}_t(\cinst{J})$ is computable in polynomial time w.r.t.\ $|\cinst{J}|$.
	\end{enumerate}
\end{theorem}

The theorem above implies that the approximation algorithm provides so-called \emph{correctness guarantees} (Item 1 of the theorem), i.e., the algorithm always constructs a subset of the conditional certain answers, and thus, only returns correct answers. This is the standard notion for measuring the quality of the set of approximate answers these algorithms are able to compute, in the context of querying incomplete databases---e.g., see~\cite{Libkin16,GuagliardoL16,ConsoleGL16}.
To the best of our knowledge, none of the existing approximation algorithms from the literature provide other kinds of theoretical guarantees, e.g., w.r.t.\ to ``how complete'' the set of approximate answers is.
%

From the result above, Theorem~\ref{thm:ptime-universal}, and Definition~\ref{def:cond-chase-seq}, we obtain the following crucial result.

\begin{corollary}
	Consider a TGD-only weakly-acyclic setting $\Set$. For every source instance $I$ of $\Set$, an approximate conditional solution $\cinst{J}$ of $I$ w.r.t.\ $\Set$ can be constructed in polynomial time, and for every query $Q$, $\conditional{Q}_t$ is such that
	\begin{enumerate}[align=parleft,left=0pt..1em]
		\item $\conditional{Q}_t(\cinst{J}) \subseteq \ccert{\cinst{J}}{Q} \subseteq \scert{I}{Q}{\Set}$;
		\item if $Q$ is positive,  $\conditional{Q}_t(\cinst{J})$ = $\ccert{\cinst{J}}{Q}$ = $\scert{I}{Q}{\Set}$;
		\item $\conditional{Q}_t(\cinst{J})$ is computable in polynomial time w.r.t.\ $|\cinst{J}|$.
	\end{enumerate}
	
\end{corollary}

%
%

\section{Dealing with EGDs}\label{sec:egds}
We now show how to deal with weakly-acyclic settings with EGDs, when it comes to construct approximate conditional solutions. 

Consider a weakly-acyclic data exchange setting $\Set = \des$ and a source instance $I$. We assume that $\ssol{I,\Set} \neq \emptyset$. Checking whether $\ssol{I,\Set} = \emptyset$ is feasible in polynomial time, for weakly-acyclic settings (Theorem~\ref{thm:exsol-wa}), and if $\ssol{I,\Set}$ is empty, no approximate conditional solution can be constructed. 

The goal is to first construct an approximate conditional solution $\cinst{J}$ for the data exchange setting $\Set^\exists$ obtained from $\Set$ by removing the set $\Sigma_E$ of all EGDs from $\TTGD$.
Then, we show that for every query $Q$, the EGDs in $\Sigma_E$ can be embedded in $Q$, obtaining a query $Q'$, in such a way that
$$\ccert{\cinst{J}}{Q'} = \bigcap\limits_{J \in \pw{\cinst{J}}\text{ and } J \text{ satisfies } \Sigma_E} Q(J).$$
As we will see, this will imply that $\ccert{\cinst{J}}{Q'} \subseteq \scert{I}{Q}{\Set}$. 

Thus, modulo a rewriting of $Q$, we can exploit $\cinst{J}$ to compute an approximation of the supported certain answers of $Q$. 
Despite our efforts, we were not able to prove that $Q'$ is also such that $\ccert{\cinst{J}}{Q'} = \scert{I}{Q}{\Set}$, when $Q$ is positive. It is an open question that we hope to answer in a future work.

In what follows, for a data exchange setting $\Set$, $\Set^{\exists}$ denotes the setting obtained from $\Set$ by removing the set $\Sigma_E$ of all EGDs from $\TTGD$.

\begin{lemma}\label{lem:qrew}
	Consider a weakly-acyclic data exchange setting $\Set = \des$, and assume $I$ is a source instance of $\Set$ such that $\ssol{I,\Set} \neq \emptyset$. Moreover, let $\cinst{J}$ be an approximate conditional solution of $I$ w.r.t.\ $\Set^\exists$. Then, for every query $Q$, there exists a query $Q'$, which depends only on $Q$ and the set of EGDs $\Sigma_E$ in $\Set$, such that
	\begin{flalign*}
		\ccert{\cinst{J}}{Q'} =\bigcap\limits_{J \in \pw{\cinst{J}}\text{ and } J \text{ satisfies } \Sigma_E} Q(J).
	\end{flalign*}
\end{lemma}
\begin{proof} Let $k = \ar(Q)$. The goal is to construct, for a given query $Q$, a query $Q'$ such that, for every target instance $J$, whenever all the EGDs in $\Sigma_E$ are satisfied by $J$, then $Q'(J) = Q(J)$, and $Q'(J) = \mathcal{C}^k$ otherwise, where $\mathcal{C}$ is the set of all constants occurring in $J$, $\Set$, and $Q$. That is, if $J$ does not satisfy $\Sigma_E$, the query $Q'$ outputs every possible tuple of length $k$, using constants from $J$, $\Set$ and $Q$.  Clearly, if $Q'$ enjoys the above property, the claim will follow immediately.
	We now explain how the query $Q'$ can be constructed, starting from $Q$ and $\Sigma_E$. 
	The query $Q'$ is made of two subqueries, that are put together via a union. That is:
	$$Q' = Q_1 \vee Q_2.$$
	$Q_1$ is such that for every target instance $J$, if $J$ satisfies $\Sigma_E$, then $Q_1(J) = Q(J)$, and $Q_1(J) = \emptyset$, otherwise.
	On the other hand, $Q_2$ is such that for every target instance $J$, if $J$ satisfies $\Sigma_E$, then $Q_2(J) = \emptyset$, and $Q_2(J) = \mathcal{C}^{k}$, otherwise. It remains to show how $Q_1$ and $Q_2$ are constructed. 
	For each EGD $\eta \in \Sigma_E$, we let $Q_{\eta}$ be the boolean query such that for every instance $J$ over $\Trg$, $Q_\eta(J) = \{()\}$, if $\eta$ satisfies $J$, and $Q_\eta(J) = \emptyset$, otherwise.
	Furthermore, we use $Q^{\neg}_{\eta}$ to denote the complement of $Q_{\eta}$, that is $Q^\neg_\eta(J) = \{()\}$ iff $Q_\eta(J) = \emptyset$.
	All the above queries can be easily written in FO.
	Finally, we let $Q_{\mathsf{dom}}$ be the query of arity $k$, such that, for every target instance $J$, $Q_{\mathsf{dom}}(J)$ is the set of all tuples of length $k$ over the constants in $J$, $\Set$ and $Q$. The above query can be encoded with a UCQ.
	Then, we have
	$$ Q_1(x_1,\ldots,x_k) = Q(x_1,\ldots,x_k) \wedge \bigwedge\limits_{\eta \in \Sigma_E} Q_\eta,$$
	and
	$$ Q_2(x_1,\ldots,x_k) = Q_{\mathsf{dom}}(x_1,\ldots,x_k) \wedge \bigvee\limits_{\eta \in \Sigma_E} Q^\neg_\eta.$$
	By construction, $Q_1(J) = Q(J)$ if $J$ satisfies $\Sigma_E$ and $Q_1(J) = \emptyset$, otherwise, and $Q_2(J) = \emptyset$, if $J$ satisfies $\Sigma_E$, and $Q_2(J) = \mathcal{C}^{k}$, otherwise.
\end{proof}

From the result above, and from the fact that the supported solutions of a data exchange setting $\Set$ correspond to the supported solutions of $\Set^\exists$ that also satisfy the EGDs of $\Set$, we obtain the main result of this section.

\begin{theorem}
	Consider a weakly-acyclic data exchange setting $\Set = \des$, and assume $I$ is a source instance of $\Set$ such that $\ssol{I,\Set} \neq \emptyset$. Moreover, let $\cinst{J}$ be an approximate conditional solution of $I$ w.r.t.\ $\Set^\exists$. Then, for every query $Q$, there exists a query $Q'$, which depends only on $Q$ and the set of EGDs $\Sigma_E$ in $\Set$, such that
	\begin{flalign*}
		&& \ccert{\cinst{J}}{Q'} \subseteq \scert{I}{Q}{\Set}. && 
	\end{flalign*}
\end{theorem}
\begin{proof}
From Lemma~\ref{lem:qrew}, there exists a query $Q'$, depending only on $Q$ and $\Sigma_E$, such that
\begin{equation}\label{eq:eq}
	\ccert{\cinst{J}}{Q'} =\bigcap\limits_{J \in \pw{\cinst{J}}\text{ and } J \text{ satisfies } \Sigma_E} Q(J).
\end{equation}
From the definition of approximate conditional solution, we have that $\ssol{I,\Set^\exists} \subseteq \pw{\cinst{J}}$. Moreover, by definition of supported solution, $\ssol{I,\Set} = \{J \in \ssol{I,\Set} \mid J \text{ satisfies } \Sigma_E\}$. Hence, $\ssol{I,\Set} \subseteq \{ J \in \pw{\cinst{J}} \mid J \text{ satisfies } \Sigma_E\}$. The latter inclusion and equation~\ref{eq:eq} let us conclude that
$\ccert{\cinst{J}}{Q'} \subseteq \scert{I}{Q}{\Set}$.
\end{proof}

The above results tell us that we can still materialize a target instance, even for weakly-acyclic settings that allow for EGDs. Moreover, modulo a rewriting of the query $Q$, the constructed target instance allows for the construction of a subset of supported certain answers of $Q$.


\section{Connections with Other Work and Next Steps}
Conditional instances and, more in general, incomplete databases, have already been employed in the context of data exchange. However, in most of previous work, incomplete databases are used to encode source and target instances with incomplete information.
In~\cite{Arenas0R13}, the authors extend the standard data exchange framework by allowing source and target instances to be incomplete databases, encoded via some representation system, such as conditional instances. There, the main goal is to study the semantics of data exchange under the assumption that the source and target instances can be incomplete.
In contrast, in our work, we focus on the classical data exchange setting, where source and target instances are standard (complete) databases. Here we employ incomplete databases, in particular conditional instances, only as a \emph{tool} to compute the (approximate) certain answers of a query over our set of supported solutions, which are standard databases as well. 
Adapting our notion of supported solution to the setting of data exchange with incomplete instances is a non-trivial task which we will consider for future work.

In Section~\ref{sec:univ-solutions}, we have seen how a conditional extension of the chase procedure, working on a normalized form of TGDs, can be employed to compute in polynomial time, for weakly-acyclic settings, an approximate conditional solution.
A similar normal form to the one we employ in our paper is presented in~\cite{GheerbrantS19}. However, in that work, the normal form is applied to \emph{queries}, and the goal is to compute so-called best answers of UCQs over incomplete databases, while in our case, we employ a normal form for \emph{TGDs}, which we then use to simplify the definition of the conditional chase.
Finally, the idea of extending the chase procedure with conditional TGD applications is not new and has been explored in previous work. In particular, the work of~\cite{GrahneO11} introduces a conditional version of the chase procedure which is similar to ours. The main difference is that the conditional chase of~\cite{GrahneO11} is much simpler, since it is an extension of the simplest variant of the chase algorithm, called oblivious chase, while ours can be seen as an extension of the more refined semi-oblivious (a.k.a.\ skolem) chase (see., e.g.,~\cite{CalauttiGP15,GrOn18,CalauttiP19,CalauttiP21,CalauttiGP22} for more details). For this reason, it is not difficult to show that when considering weakly-acyclic settings, the conditional chase of~\cite{GrahneO11} is not guaranteed to terminate, while termination for weakly-acyclic settings is a crucial property for our purposes, since we need to be able to construct a finite conditional instance in this case.

The problem of dealing with non-monotonic queries has been investigated beyond data exchange, as for example for Ontology-Mediated Query Answering (OMQA). In this setting, we are given an instance (database) $D$, an ontology $\Sigma$ encoded in some logical formalism (e.g., via TGDs), and a query $Q(\x)$, and the goal is to compute all the certain answers of $Q(\x)$ w.r.t.\ $D$ and $\Sigma$, i.e., the tuples that are answers to $Q$ in \emph{every model} of the logical theory $D \cup \Sigma$. A relevant work in this scenario is the one in~\cite{Calvanese07}, where the authors define the query language EQL-Lite($\cal Q$), parametrized with a standard (positive) query language $\cal Q$ (e.g., UCQs), and supports a limited form of negation. In particular, an expression $\psi$ in EQL-Lite($\cal Q$) is of the form
$\psi := \mathbb{K}\, \rho \mid \psi_1 \wedge \psi_2 \mid \neg \psi_1 \mid \exists x\, \psi_1$,
where $\rho \in \cal Q$, and $\psi_1,\psi_2$ are EQL-Lite($\cal Q$) expressions. 

Here, the epistemic operator $\mathbb{K}\,$ is applied to expressions $\rho \in \cal Q$ and returns the certain answers of $\rho$ w.r.t.\ the input database $D$ and the ontology $\Sigma$. The main instantiation of EQL-Lite that the authors study is EQL-Lite(UCQ), i.e., where $\cal Q$ coincides with the set of all UCQs.

From the above definition, we observe that negation is applied only to (a combination of) the certain answers of positive queries. This gives a semantics to negation that fundamentally differs from ours, as illustrated in the following example.

Consider the data exchange setting $\Set = \des$, where $\Src$ stores employees of a company in the unary relation $\Emp$.
The target schema $\Trg$ contains a unary relation $\Emp'$ storing employees, the ternary relation $\Address$ assigning to each employee her work and home address, and the unary relation $\WorkFromHome$, storing employees working from home. Assume we have $\STTGD = \{\rho_1 = \Emp(x) \rightarrow \Emp'(x),\rho_2=\Emp(x) \rightarrow \exists z\, \exists w\, \Address(x,z,w)\}$ and $\TTGD = \{\rho_3=\Address(x,y,y) \rightarrow \WorkFromHome(x)\}$.

	The above setting copies employees from the source to the target via the TGD $\rho_1$, while the TGD $\rho_2$ states that each employee must have a work and home address, denoted via the existential variables $z$ and $w$, respectively. Finally, the TGD $\rho_3$ states that if the work and home address of an employee coincide, then this employee works from home.
	
	Assume the source instance is $I = \{\Emp(\rel{john})\}$,
	and let $Q$ be the query asking for all employees who do not work from home, i.e., $Q(x) = \Emp'(x) \wedge \neg \WorkFromHome(x)$.
	
	According to~\cite{Calvanese07}, the query $Q$ corresponds to the EQL-Lite(UCQ) expression
	$Q'(x) = \mathbb{K}\, \Emp'(x) \wedge \neg \mathbb{K}\, \WorkFromHome(x)$.
	Letting $D = I$, and $\Sigma = \STTGD \cup \TTGD$, roughly, the above means that an empolyee is an answer to the query $Q'$ if she is present in \emph{all} models of $D \cup \Sigma$ and such that there is \emph{at least one model} in which the employee does not work from home. Under this interpretation, the answer to $Q'$ is $\rel{john}$. However, under our semantics, the answer to $Q$ is empty. Hence, the fundamental difference is that negation, under EQL-Lite, is interpreted as negating classical certain answering, and thus an expression $\neg \mathbb{K}\, \psi$ is ``satisfied''  when at least one model/solution does not entail $\psi$, while in our case, we consider the given query as a whole, and require it to be satisfied in \emph{every} valid solution.

We conclude by discussing avenues for further research. First, we would like to extend the conditional chase to weakly-acyclic settings with EGDs, and identify relevant data exchange settings for which computing the supported certain answers is tractable.
Moreover, we would like to identify other quality measures of our approximation algorithm using techniques such as the ones introduced in~\cite{Libkin18}. We also plan to experimentally evaluate both our translation to logic programs for computing exact answers, as well as our materialization-based approaches for computing approximate answers by means of a dedicated benchmark, as done e.g., in the context of approximate consistent query answering~\cite{CalauttiCP21}.

To conclude, we mention that explaining query answering has recently drawn considerable attention under existential rule languages (e.g., see \cite{ijcai22,CeylanLMMV21,CeylanLMMV20,LukasiewiczMM20,CeylanLMV19}), and knowledge representation in general (e.g., in the context of argumentation~\cite{AlfanoCGPT20}). Hence, an interesting direction for future work is to address such issues in our setting. Also, it would be interesting to account for user preferences when answering queries, as recently done in \cite{aij22b} for ontology-mediated queries.

%
%
%

\bibliographystyle{acmtrans}
\bibliography{refs}

\end{document}